\documentclass[12pt,12pt,onecolumn,draftclsnofoot]{IEEEtran}
\usepackage[latin9]{inputenc}
\usepackage{color}
\usepackage{textcomp}
\usepackage{mathrsfs}
\usepackage{amsmath}
\usepackage{amsthm}
\usepackage{amssymb}
\usepackage{graphicx}

\makeatletter

\newcommand{\lyxmathsym}[1]{\ifmmode\begingroup\def\b@ld{bold}
  \text{\ifx\math@version\b@ld\bfseries\fi#1}\endgroup\else#1\fi}

\theoremstyle{remark}
\newtheorem{rem}{Remark}
\theoremstyle{plain}
	  \newtheorem{lem}{Lemma}
\theoremstyle{plain}
	  \newtheorem{prop}{Proposition}
\theoremstyle{plain}
		\newtheorem{thm}{Theorem}

\ifCLASSINFOpdf
\else
\fi
\usepackage{caption}
\usepackage{cite}
\renewcommand{\maketag@@@}[1]{\hbox{\m@th\normalsize\normalfont#1}}

\usepackage{etoolbox}
\usepackage{xstring}\usepackage{mfirstuc}\usepackage{textcase}\usepackage[acronym]{glossaries}
\newcommand{\newac}{\newacronym}
\newcommand{\ac}{\gls}
\newcommand{\Ac}{\Gls}
\newcommand{\acpl}{\glspl}

\makeglossaries
\newac{isac}{ISAC}{integrated sensing and communication}
\newac{an}{AN}{artificial noise}
\newac{csi}{CSI}{channel state information}
\newac{csit}{CSIT}{channel state information at transmitter}
\newac{fp}{FP}{fractional planning}
\newac{bme}{BME}{beampattern matching error}
\newac{ula}{ULA}{uniform linear array}
\newac{cu}{CU}{communication user}
\newac{uavs}{UAVs}{unmanned aerial vehicles}
\newac{sinr}{SINR}{signal-to-interference-plus-noise ratio}
\newac{sdr}{SDR}{semi-definite relaxation}
\newac{1d}{1D}{one-dimensional}
\newac{zf}{ZF}{zero-forcing}
\newac{aod}{AoD}{angle of departure}
\newac{los}{LoS}{line-of-sight}
\newac{nlos}{NLoS}{non-line-of-sight}
\newac{bs}{BS}{base station}
\newac{bss}{BSs}{base stations}
\newac{cus}{CUs}{communication users}
\newac{snr}{SNR}{signal-to-noise ratio}
\newac{evd}{EVD}{eigenvalue decomposition}
\newac{qsdp}{QSDP}{quadratic semidefinite programing}
\newac{dof}{DoF}{degrees-of-freedom}
\newac{soc}{SOC}{second order cone}
\newac{qos}{QoS}{quality-of-service}
\newac{lmi}{LMI}{linear matrix inequality}
\newac{bti}{BTI}{Bernstein-type inequality}

\makeatother

\begin{document}
\title{Robust Transmit Beamforming for Secure Integrated Sensing and Communication}
\author{Zixiang Ren, Ling Qiu, Jie Xu, and Derrick Wing Kwan Ng\thanks{Part of this paper has been presented at the IEEE International Conference
on Communications (ICC), Seoul, South Korea, 16-20 May 2022 \cite{ren2021optimal}.}\thanks{Z. Ren is with Key Laboratory of Wireless-Optical Communications,
Chinese Academy of Sciences, School of Information Science and Technology,
University of Science and Technology of China, Hefei 230027, China,
and the School of Science and Engineering (SSE) and the Future Network
of Intelligence Institute (FNii), The Chinese University of Hong Kong
(Shenzhen), Shenzhen 518172, China (e-mail: rzx66@mail.ustc.edu.cn).}\thanks{L. Qiu is with Key Laboratory of Wireless-Optical Communications,
Chinese Academy of Sciences, School of Information Science and Technology,
University of Science and Technology of China, Hefei 230027, China,
(e-mail: lqiu@ustc.edu.cn).}\thanks{J. Xu is with the SSE and the FNii, The Chinese University of Hong
Kong (Shenzhen), Shenzhen 518172, China (e-mail: xujie@cuhk.edu.cn).}\thanks{D. W. K. Ng is with the University of New South Wales, Sydney, NSW
2052, Australia (e-mail: w.k.ng@unsw.edu.au).}\thanks{L. Qiu and J. Xu are the corresponding authors.}\vspace{-1.5cm}}
\maketitle
\begin{abstract}
This paper studies a downlink secure integrated sensing and communication
(ISAC) system, in which a multi-antenna base station (BS) transmits
confidential messages to a single-antenna communication user (CU)
while performing sensing on targets that may act as suspicious eavesdroppers.
To ensure the quality of target sensing while preventing their potential
eavesdropping, the BS combines the transmit confidential information
signals with additional dedicated sensing signals, which play a dual
role of artificial noise (AN) for degrading the qualities of eavesdropping
channels. Under this setup, we jointly design the transmit information
and sensing beamforming, with the objective of minimizing the weighted
sum of beampattern matching errors and cross-correlation patterns
for sensing subject to secure communication constraints. The robust
design takes into account the channel state information (CSI) imperfectness
of the eavesdroppers in two practical CSI error scenarios. First,
we consider the scenario with bounded CSI errors of eavesdroppers,
in which the worst-case secrecy rate constraint is adopted to ensure
secure communication performance. In this scenario, we present the
optimal solution to the worst-case secrecy rate constrained sensing
beampattern optimization problem, by adopting the techniques of S-procedure,
semi-definite relaxation (SDR), and a one-dimensional (1D) search,
for which the tightness of the SDR is rigorously proved. Next, we
consider the scenario with Gaussian CSI errors of eavesdroppers, in
which the secrecy outage probability constraint is adopted. In this
scenario, we present an efficient algorithm to solve the more challenging
secrecy outage-constrained sensing beampattern optimization problem,
by exploiting the convex restriction technique based on the Bernstein-type
inequality, together with the SDR and 1D search. Finally, numerical
results show that the proposed designs can properly adjust the information
and sensing beams to balance the tradeoffs among communicating with
CU, sensing targets, and confusing eavesdroppers, so as to achieve
desirable sensing transmit beampatterns while ensuring the CU's secrecy
requirements for the two scenarios. 
\end{abstract}

\begin{IEEEkeywords}
Integrated sensing and communication (ISAC), physical layer security,
artificial noise (AN), transmit beamforming, secrecy rate, optimization. 
\end{IEEEkeywords}

\section{Introduction}

\Ac{isac} has been widely recognized as one of the key technologies
to realize future beyond fifth-generation (B5G) and sixth-generation
(6G) wireless networks \cite{LiuMasJ20,RahLusJ20,liu2021integrated},
in which sensing is integrated as new functionality to enable environment-aware
intelligent applications, such as autonomous driving, industrial automation,
and \ac{uavs} \cite{wUxUj21}. Different from conventional wireless
sensing and communication systems that are separately designed to
compete with each other on spectrum, energy, and hardware resources,
ISAC unifies the two functionalities into a joint design to exploit
their mutual benefit via seamless collaboration \cite{liu2021integrated}.
By allowing the dual use of wireless infrastructures and scarce spectrum/power
resources for performing communication and sensing simultaneously,
ISAC is expected to achieve enhanced sensing and communication performance
at a reduced cost.

Various research efforts have been devoted in the recent literature
to unleash the ISAC potential from different aspects, such as multi-antenna
beamforming \cite{LiuHuangNirJ20,hua2021optimal}, waveform design
\cite{LiuZhouJ18}, and resource allocation \cite{ShiWangJ19,liu2021integrated}.
Among others, multi-antenna beamforming is particularly appealing,
in which multiple signal-carrying beams can be steered towards both
the desired \acpl{cu} and sensing targets for simultaneous multi-user
communications and multi-target radar sensing over the same frequency
band. For instance, the authors in \cite{LiuZhouJ18} proposed to
reuse the multiple users' information signals for radar sensing, in
which the transmit multi-user information beamforming vectors are
optimized to minimize the inter-user interference at the \acpl{cu},
while ensuring the beampattern requirements for multi-target sensing.
However, when the number of \acpl{cu} is less than that of transmit
antennas, the exploitation of information beams only may lead to sensing
beampattern distortion, due to the limited total \ac{dof} utilized
in the information beams. To resolve this issue, the authors in \cite{LiuHuangNirJ20,hua2021optimal}
proposed to additionally send the dedicated sensing signal beams together
with the information beams such that the full spatial \ac{dof}
can be exploited for multi-target radar sensing. In particular, the
information and sensing beams were jointly optimized in \cite{LiuHuangNirJ20,hua2021optimal}
to achieve desired transmit beampattern for radar sensing, while ensuring
the individual \ac{sinr} constraints at the CUs.

Recently, security has become increasingly important for wireless
networks due to the inherent broadcast natures, while the emergence
of \ac{isac} introduces new communication security issues. In particular,
as the information-bearing signals are reused for radar sensing, the
optimized transmit information beams in ISAC would be steered towards
sensing targets to ensure sensing performance. This, however, may
lead to potential information leakage issues, as in practice the sensing
targets can be suspicious eavesdroppers (e.g., suspicious UAV targets)
that may try to intercept the messages sent from the \ac{isac}
transmitter\textcolor{blue}{{} }\cite{DelAnaJ18,SuLiuChrJ21,WeiM22}.
To deal with the new security concern in \ac{isac}, physical layer
security has emerged as a viable solution to prevent potential eavesdropping
from the sensing targets. Different from conventional cryptography-based
security technologies, physical layer security aims to provide perfectly
secure data delivery from the information-theoretic perspective, by
exploiting the difference in physical properties between the legitimate
communication channels and the eavesdropping channels (e.g., noise,
interference, and fading). In practice, the secrecy rate and secrecy
capacity have been employed as key performance metrics in physical
layer security \cite{wyner1975wire,GopPraJ08,ChenKwanJ17}. Specifically,
the secrecy capacity is defined as the difference between the capacity
of the legitimate channels and that of the eavesdropping channels,
which indicates the maximum achievable data rate that can be securely
transmitted over the wireless channels, provided that the eavesdroppers
cannot intercept any information \cite{GoeSatJ08}. To enhance the
secrecy performance, \ac{an} has emerged as a promising technique
in physical layer security \cite{WuKhiJ18}, in which the transmitters
can inject AN superimposing with the confidential information concurrently
for protecting the confidential message transmission via confusing
the eavesdroppers. Furthermore, in secure ISAC systems, AN can also
be exploited as sensing signals, with a dual role of jamming the eavesdroppers
and sensing the targets. As such, the interplay between ISAC and physical
layer security enables a multi-function wireless network integrating
sensing, communication, and security \cite{WeiM22}.

In general, the availability of \ac{csi} at the transmitter is
crucial for the implementation of physical layer security and multi-antenna
transmit beamforming \cite{WuKhiJ18,wangbozhaoJ19}. Conventionally,
to obtain the \ac{csi} of a \ac{cu}, the transmitter can either
exploit the CU's forward-link channel estimation and feedback in frequency
division duplex (FDD) systems, or perform the reverse-link channel
estimation based on the channel reciprocity in time division duplex
(TDD) systems. To acquire the CSI of eavesdroppers, the transmitter
may overhear their emitted radio signals for channel estimation if
the eavesdroppers are active in transmission. In general, however,
it is more difficult to obtain accurate CSI of eavesdroppers. As such,
how to jointly design the transmission of both information and AN/sensing
signals in the case of the imperfect CSI of eavesdroppers is an interesting
problem in secure \ac{isac} systems. 

There have been a handful of prior works investigating the secrecy
issue in \ac{isac} systems. For instance, the authors in \cite{DelAnaJ18}
studied the secure \ac{isac} system with one \ac{cu} and one
eavesdropping target. By considering the perfect \ac{csi} at the
transmitter, the transmit covariance matrices of both information
signals and \ac{an} were jointly optimized to maximize the \ac{cu}'s
secrecy rate while ensuring the received \ac{sinr} requirement
for sensing. Also, the authors in \cite{SuLiuChrJ21} considered the
secure ISAC system with multiple \acpl{cu} and one eavesdropping
target by assuming \ac{los} eavesdropping channels. For both scenarios
with perfect and imperfect \ac{csi}, the \ac{snr} at the eavesdropping
target was minimized in \cite{SuLiuChrJ21} while ensuring the \ac{cu}s'
individual \ac{sinr} and certain sensing beampattern requirements.
However, due to the non-convexity of the formulated problems, only
sub-optimal beamforming/precoding designs were obtained in these prior
works \cite{SuLiuChrJ21,DelAnaJ18} under their respective setups.
More importantly, these prior works \cite{SuLiuChrJ21,DelAnaJ18}
on secure ISAC with single target sensing only adopted the beampattern
matching mean squared error (MSE) as the sensing performance metric
but ignored the impacts of cross-correlation patterns for multi-target
sensing that are shown to be important \cite{StoPETLiJ07,LiuHuangNirJ20}.

This paper studies a secure \ac{isac} system with a multi-antenna
\ac{bs}, a single-antenna \ac{cu}, and multiple sensing targets,
in which a portion of targets are untrusted (acting as suspicious
eavesdroppers) and the remaining are trusted. Different from prior
works, e.g., \cite{SuLiuChrJ21,DelAnaJ18}, we consider the weighted
sum of the beampattern matching MSE and cross correlation patterns
as the sensing performance measure. We consider that the \ac{bs}
sends one confidential information signal beam and multiple dedicated
sensing signal beams that also play the AN role. It is assumed that
the BS has the perfect CSI of the CU but imperfect CSI of the potential
eavesdroppers subject to two different types of errors, namely the
bounded and Gaussian errors, respectively. For these two scenarios,
we jointly optimize the transmit information and sensing beamforming
vectors to minimize the weighted sum of the beampattern matching error
and cross-correlation patterns for sensing, while ensuring the secure
communication requirements. The main results are listed as follows. 
\begin{itemize}
\item First, we consider the scenario with bounded CSI errors of the eavesdroppers,
in which the minimum worst-case secrecy rate requirement is adopted
at the \ac{cu} to ensure secrecy performance. In this scenario,
the worst-case secrecy rate-constrained sensing beampattern optimization
problem is non-convex and thus difficult to solve in general. Fortunately,
we obtain the optimal solution to this problem via the techniques
of S-procedure, \ac{sdr}, and a \ac{1d} search. In particular,
we rigorously prove the tightness of \ac{sdr}, based on which the
optimal solution with rank-one information beamforming and general-rank
sensing beamforming can always be obtained via the proposed construction
approach. 
\item Next, we consider the scenario with Gaussian CSI errors of the eavesdroppers.
In this scenario, a maximum secrecy outage probability constraint
is imposed to ensure secure communication performance such that the
likelihood of the achievable secrecy rate being less than a required
level is small. The secrecy outage-constrained sensing beampattern
optimization problem, however, is more difficult to solve compared
with the worst-case counterpart, as the maximum secrecy outage constraint
does not admit a closed-form expression with respect to the transmit
beamforming vectors. To solve this problem, we first apply the convex
restriction technique based on the \ac{bti} to approximate the
original problem. Then, we solve the restricted approximate problem
via the SDR and 1D search, and finally propose an efficient construction
method to find the desired solution to the original problem with rank-one
information and general-rank sensing beamformers.
\item Finally, we provide numerical results to validate the effectiveness
of our proposed joint beamforming designs for the two scenarios with
imperfect CSI of eavesdroppers. It is shown that in our proposed solution,
the information beams are designed toward the \ac{cu} and the trusted
targets, and the sensing/AN beams are designed toward both untrusted
and trusted targets, thus maximizing the sensing performance while
ensuring the CU's secure communication requirement. It is also shown
that as compared to a conventional benchmarks, our proposed joint
beamforming design achieves better sensing performance in terms of
lower beampattern matching errors and higher accuracy for angle estimation
(by considering the practical Capon estimation algorithm \cite{StoPETLiJ07,xu2008target}
at the sensing receiver). We also illustrate the impacts of CSI errors
on the ISAC performance. It is shown that our proposed designs are
robust to the CSI errors of eavesdroppers. 
\end{itemize}

The remainder of this paper is organized as follows. Section II introduces
the system model and formulates the secrecy constrained sensing beampattern
optimization problems for the two scenarios with bounded and Gaussian
CSI errors of eavesdroppers. Sections III and IV develop efficient
algorithms for solving the two problems in the two scenarios, respectively.
Section V presents numerical results and Section VI concludes this
paper.

\textit{Notations}: Throughout this paper, vectors and matrices are
denoted by bold lower- and upper-case letters, respectively. $\mathbb{C}^{N\times M}$
denotes the space of $N\times M$ matrices with complex entries. $\mathbb{S}^{N}$
denotes the space of $N\times N$ Hermitian matrices, $\mathbb{S}_{+}^{N}$
denotes the space of $N\times N$ positive semi-definite matrices,
and $\mathbb{S}_{++}^{N}$ denotes the space of $N\times N$ positive
definite matrices. $\boldsymbol{I}$ and $\boldsymbol{0}$ represent
an identity matrix and an all-zero vector/matrix with appropriate
dimensions, respectively. For a complex-valued element $a$, $\textrm{Re}(a)$
and $\mathrm{Im}(a)$ denote its real part and imaginary part, respectively.
For a square matrix $\boldsymbol{A}$, $\textrm{Tr}(\boldsymbol{A})$
denotes its trace and $\boldsymbol{A}\succeq\boldsymbol{0}$ means
that $\boldsymbol{A}$ is positive semi-definite. For a complex arbitrary-size
matrix $\boldsymbol{B}$, $\textrm{rank}(\boldsymbol{B})$, $\boldsymbol{B}^{T}$,
$\boldsymbol{B}^{H}$, and $\boldsymbol{B}^{c}$ denote its rank,
transpose, conjugate transpose, and complex conjugate, respectively.
$\mathbb{E}(\cdot)$ denotes the stochastic expectation, $\|\cdot\|$
denotes the Euclidean norm of a vector, and $\mathcal{CN}(\boldsymbol{x},\boldsymbol{Y})$
denotes the circularly symmetric complex Gaussian (CSCG) random distribution
with mean vector $\boldsymbol{x}$ and covariance matrix $\boldsymbol{Y}$.
$(x)^{+}\triangleq\max(x,0)$.


\section{System Model and Problem Formulation}

\subsection{Signal Model}

\begin{figure}
\includegraphics[width=10cm,height=6cm]{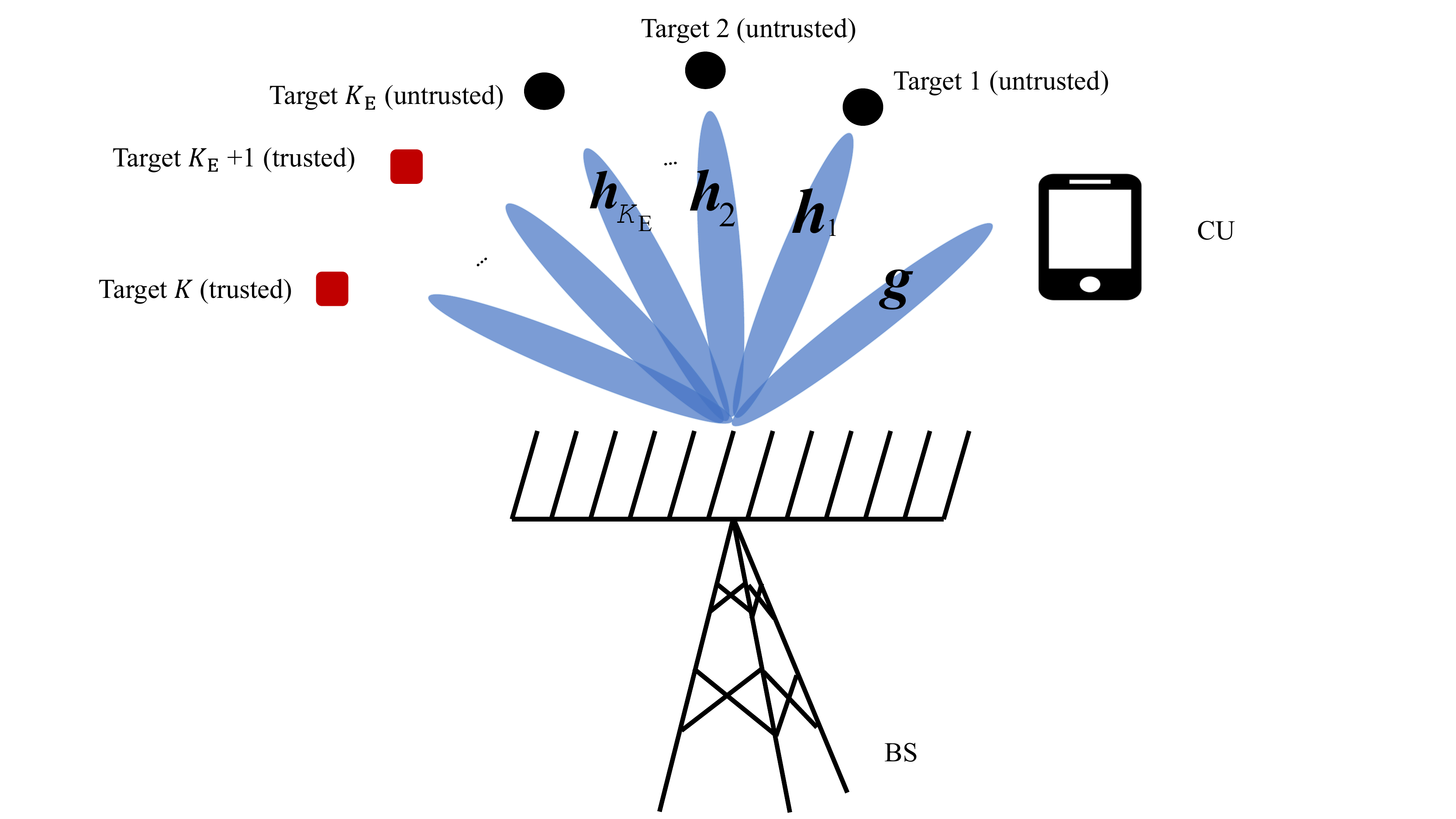}\centering\caption{\label{fig1}Illustration of the secure ISAC system with both the
trusted and untrusted targets. }
\end{figure}

We consider a secure \ac{isac} system as shown in Fig. \ref{fig1},
which consists of a BS equipped with a \ac{ula} with $N>1$ antenna
elements, a \ac{cu} with a single antenna\footnote{Our proposed design principles can be extended to the case with multiple
CUs each with multiple antennas, by accordingly modifying the secrecy
rate formulas.}, and $K$ sensing targets. Let $\mathcal{K}\triangleq\{1,...,K\}$
denote the set of targets, among which the first $K_{\mathrm{E}}$
ones (with $K_{\mathrm{E}}\le K$) are assumed to be untrusted or
suspicious eavesdroppers, denoted by set $\mathcal{K}_{\mathrm{E}}\ensuremath{\triangleq}\{1,...,K_{\mathrm{E}}\}\subseteq\mathcal{K}$,
and the remaining $K-K_{\mathrm{E}}$ targets are assumed to be trusted
ones. We consider that the BS adopts the linear transmit beamforming
to send the confidential message $s_{0}(n)\in\mathbb{C}$ to the \ac{cu},
where $s_{0}(n)\in\mathbb{C}$ is a CSCG random variable with zero
mean and unit variance, i.e., $s_{0}(n)\sim\mathcal{CN}(0,1)$, with
$n\in\{1,\dots,L\}$ denoting the symbol index. We use $\boldsymbol{w}_{0}\in\mathbb{C}^{N\times1}$
to denote the corresponding transmit information beamforming vector.
Besides the information signal $s_{0}(n)$, the BS also sends a dedicated
sensing signal or \ac{an}, denoted by vector $\boldsymbol{s}_{1}(n)\in\mathbb{C}^{N\times1}$,
to facilitate the target sensing and to confuse the eavesdropping
targets. Here, $\boldsymbol{s}_{1}(n)$ is independent from $s_{0}(n)$
and is a CSCG random vector with zero mean and covariance matrix $\boldsymbol{S}=\mathbb{E}(\boldsymbol{s}_{1}(n)\boldsymbol{s}_{1}^{H}(n))\succeq\boldsymbol{0},\textrm{i.e}.,$
$\boldsymbol{s}_{1}(n)\sim\mathcal{CN}(\boldsymbol{0},\boldsymbol{S})$.
We assume that $s_{0}(n)$ and $\boldsymbol{s}_{1}(n)$ are independent
over different symbols, $\forall n\in\{1,\dots,L\}$. Note that $\boldsymbol{S}$
is assumed to be of a general rank with $0\le m=\textrm{rank}(\boldsymbol{S})\le N$.
This corresponds to the case with $m$ linearly and statistically
independent sensing beams, each of which can be obtained via the \ac{evd}
of $\boldsymbol{S}$ \cite{wangbozhaoJ19}. As a result, the transmitted
signal by the \ac{bs}, $\boldsymbol{x}(n)\in\mathbb{C}^{N\times1}$,
is expressed as
\begin{equation}
\boldsymbol{x}(n)=\boldsymbol{w}_{0}s_{0}(n)+\boldsymbol{s}_{1}(n).
\end{equation}
Accordingly, the covariance matrix of $\boldsymbol{x}(n)$ is denoted
as
\begin{equation}
\boldsymbol{T}=\mathbb{E}(\boldsymbol{x}(n)\boldsymbol{x}^{H}(n))=\boldsymbol{W}+\boldsymbol{S},
\end{equation}
where $\boldsymbol{W}=\boldsymbol{w}_{0}\boldsymbol{w}_{0}^{H}$ with
$\boldsymbol{W}\succeq\boldsymbol{0}$ and $\textrm{rank}(\boldsymbol{W})\leq1$.
We consider that the BS is subject to the per-antenna transmit power
constraints, i.e., each antenna at the BS is subject to the maximum
transmit power budget $Q$. We define $\boldsymbol{E}_{i}\in\mathbb{C}^{N\times N}$
as a zero matrix except the $i$-th diagonal element being $1$. In
this case, we have
\begin{equation}
\mathrm{Tr}(\boldsymbol{E}_{i}\boldsymbol{T})=\mathrm{Tr}(\boldsymbol{E}_{i}(\boldsymbol{W}+\boldsymbol{S}))=Q,\forall i\in\{1,\dots,N\},\label{eq:power constraint}
\end{equation}
in which the full power transmission is adopted at each transmit antenna
of the BS to maximize the sensing performance. Similar per-antenna
full power transmission has been widely considered in the MIMO radar
literature \cite{LiStoJ07,StoPETLiJ07}.

\subsection{Secure Communication}

Then, we consider the secure communication model. We consider a quasi-static
channel model, in which the wireless channels remain static over the
time block of our interest, but may change from one block to another
\cite{LiuHuangNirJ20,SuLiuChrJ21}. Let $\boldsymbol{g}\in\mathbb{C}^{N\times1}$
denote the channel vector from the BS to the \ac{cu}. Accordingly,
the received signal at the \ac{cu} is expressed as\setlength{\abovedisplayskip}{1pt}
\begin{equation}
r_{0}(n)=\boldsymbol{g}^{H}\boldsymbol{w}_{0}s_{0}(n)+\boldsymbol{g}^{H}\boldsymbol{s}_{1}(n)+z_{0}(n),\label{eq:Received signal at CU}
\end{equation}
where $z_{0}(n)\sim\mathcal{CN}(0,\sigma_{0}^{2})$ denotes the additive
white Gaussian noise (AWGN) at the \ac{cu} receiver and $\sigma_{0}^{2}$
denotes the corresponding noise power. Based on (\ref{eq:Received signal at CU}),
the received \ac{sinr} at the \ac{cu} is\setlength{\abovedisplayskip}{1pt}
\begin{equation}
\gamma_{0}(\boldsymbol{W},\boldsymbol{S},\boldsymbol{g})=\frac{\boldsymbol{g}^{H}\boldsymbol{W}\boldsymbol{g}}{\boldsymbol{g}^{H}\boldsymbol{S}\boldsymbol{g}+\sigma_{0}^{2}}.\label{eq:SINR at CU}
\end{equation}

\setlength{\belowdisplayskip}{1pt}Furthermore, by letting $\boldsymbol{h}_{k}\in\mathbb{C}^{N\times1}$
denote the channel vector from the \ac{bs} to untrusted target
$k\in\mathcal{K}_{\mathrm{E}}$, the received signal at untrusted
target $k\in\mathcal{K}_{\mathrm{E}}$ is denoted as\setlength{\abovedisplayskip}{1pt}
\begin{equation}
r_{k}(n)=\boldsymbol{h}_{k}^{H}\boldsymbol{w}_{0}s_{0}(n)+\boldsymbol{h}_{k}^{H}\boldsymbol{s}_{1}(n)+z_{k}(n),\label{eq:Received signal at target}
\end{equation}
\setlength{\belowdisplayskip}{1pt}where $z_{k}(n)\sim\mathcal{CN}(0,\sigma_{k}^{2})$
denotes the AWGN at the receiver of untrusted target $k\in\mathcal{K}_{\mathrm{E}}$
and $\sigma_{k}^{2}$ denotes the corresponding noise power. Thus,
the \ac{sinr} at untrusted target $k\in\mathcal{K}_{\mathrm{E}}$
is 
\begin{equation}
\gamma_{k}(\boldsymbol{W},\boldsymbol{S},\boldsymbol{h}_{k})=\frac{\boldsymbol{h}_{k}^{H}\boldsymbol{W}\boldsymbol{h}_{k}}{\boldsymbol{h}_{k}^{H}\boldsymbol{S}\boldsymbol{h}_{k}+\sigma_{k}^{2}}.
\end{equation}

As such, under given $\boldsymbol{g}$ and $\{\boldsymbol{h}_{k}\}$,
the achievable secrecy rate at the \ac{cu} (in bits-per-second-per-Hertz,
bps/Hz) is given by \cite{GopPraJ08} 
\begin{equation}
R(\boldsymbol{W},\boldsymbol{S})=\underset{k\in\mathcal{K}_{\mathrm{E}}}{\min}\Big(\log_{2}\big(1+\gamma_{0}(\boldsymbol{W},\boldsymbol{S},\boldsymbol{g})\big)-\log_{2}\big(1+\gamma_{k}(\boldsymbol{W},\boldsymbol{S},\boldsymbol{h}_{k})\big)\Big)^{+}.\label{eq:rate}
\end{equation}

It is assumed that the BS has perfect CSI of the CU but imperfect
CSI of the eavesdroppers. Let $\hat{\boldsymbol{h}}_{k}$ denote the
estimated CSI of eavesdropping target $k\in\mathcal{K}_{\mathrm{E}}$
at the BS. The CSI estimation error vector of eavesdropping target
$k\in\mathcal{K}_{\mathrm{E}}$ is expressed as\setlength{\abovedisplayskip}{-1pt}
\begin{equation}
\boldsymbol{e}_{k}=\boldsymbol{h}_{k}-\hat{\boldsymbol{h}}_{k},\forall k\in\mathcal{K}_{\mathrm{E}}.\label{eq:error}
\end{equation}
It is assumed that the error vector $\boldsymbol{e}_{k}$ is independent
from $\hat{\boldsymbol{h}}_{k}$, $\forall k\in\mathcal{K}_{\mathrm{E}}$.
In particular, we consider two practical scenarios with bounded and
Gaussian CSI errors for $\{\boldsymbol{e}_{k}\}$, respectively.

\subsubsection{Bounded CSI Errors}

In the first scenario, the CSI of each eavesdropper is subject to
a bounded error, i.e.,\setlength{\abovedisplayskip}{-1pt} 
\begin{equation}
\|\boldsymbol{e}_{k}\|\leq\varepsilon_{k},\forall k\in\mathcal{K}_{\mathrm{E}},
\end{equation}
\setlength{\belowdisplayskip}{-1pt}where $\varepsilon_{k}$ denotes
the maximum threshold of the bounded CSI error. Under given $\boldsymbol{g}$
and distribution of $\{\boldsymbol{h}_{k}\}$, the \textit{worst-case
}secrecy rate \cite{huang2011robust} is employed as the secure communication
metric in this scenario, i.e., 
\begin{equation}
R_{\mathrm{w}}(\boldsymbol{W},\boldsymbol{S})=\underset{k\in\mathcal{K}_{\mathrm{E}}}{\min}\Big(\log_{2}\big(1+\gamma_{0}(\boldsymbol{W},\boldsymbol{S},\boldsymbol{g})\big)-\underset{\|\boldsymbol{e}_{k}\|\leq\varepsilon_{k}}{\max}\log_{2}\big(1+\gamma_{k}(\boldsymbol{W},\boldsymbol{S},\boldsymbol{h}_{k})\big)\Big)^{+}.\label{eqn:worst_case:secrecy}
\end{equation}

\subsubsection{Gaussian CSI Errors}

In the second scenario, the CSI of the eavesdroppers is subject to
Gaussian errors, i.e., $\boldsymbol{e}_{k}\sim\mathcal{CN}(\boldsymbol{0},\boldsymbol{C}_{k})$,
$\forall k\in\mathcal{K}_{\mathrm{E}}$, where $\boldsymbol{C}_{k}\in\mathbb{S}_{+}^{N}$
denotes the covariance matrix of CSI error for eavesdropping target
$k\in\mathcal{K}_{\mathrm{E}}$. For notational convenience, we express
(\ref{eq:error}) as 
\begin{equation}
\boldsymbol{h}_{k}=\hat{\boldsymbol{h}}_{k}+\boldsymbol{C}_{k}^{\frac{1}{2}}\boldsymbol{u}_{k},\forall k\in\mathcal{K}_{\mathrm{E}},\label{eq:Probabilistic error model}
\end{equation}
where $\boldsymbol{u}_{k}$ denotes independent CSCG random vectors
with mean zero and variance $\boldsymbol{I}$, i.e., $\boldsymbol{u}_{k}\sim\mathcal{CN}(\boldsymbol{0},\boldsymbol{I})$.
It is assumed that $\boldsymbol{u}_{k}$ is independent over different
$k\in\mathcal{K}_{\mathrm{E}}$. Thus, $\boldsymbol{h}_{k}$ is independent
of $\boldsymbol{h}_{l}$, for any $k,l\in\mathcal{K}_{\mathrm{E}},k\neq l$
\cite{li2013safe}.\textcolor{blue}{{} }With given $\boldsymbol{g}$
and the Gaussian CSI error model in (\ref{eq:Probabilistic error model}),
we adopt the secrecy outage probability as the secure communication
performance measure, which is defined as the probability of the achievable
secrecy rate in (\ref{eq:rate}) being lower than a given threshold
$R_{0}$, i.e., 
\begin{equation}
\mathrm{Pr}(R(\boldsymbol{W},\boldsymbol{S})<R_{0}).\label{eqn:outage:probability}
\end{equation}
In \eqref{eqn:outage:probability}, the probability is taken over
the randomness of $\{\boldsymbol{u}_{k}\}$ based on (\ref{eq:Probabilistic error model}). 

\subsection{Radar Sensing}

Next, we consider the monostatic radar sensing at the BS \cite{StoPETLiJ07,LiuHuangNirJ20}.
As the BS has the complete knowledge of the transmitted information
waveform, the information message $s_{0}(n)$ can be exploited together
with the dedicated sensing signal $\boldsymbol{s}_{1}(n)$ for target
sensing. We assume that the targets are in the far-field region, each
of which can be viewed as a point target with angle $\theta_{k}$,
$k\in\mathcal{K}$. In this case, supposing that the BS is equipped
with the same number of $N$ antennas for signal reception, the received
echo signal $\boldsymbol{y}(n)\in\mathbb{C}^{N\times1}$ at the BS
is denoted as \cite{StoPETLiJ07}\vspace{-2pt}
\begin{equation}
\boldsymbol{y}(n)=\sum_{k=1}^{K}\beta_{k}\boldsymbol{a}^{c}(\theta_{k})\boldsymbol{a}^{H}(\theta_{k})\boldsymbol{x}(n)+\boldsymbol{z}(n),\label{eq:reflected signal}
\end{equation}
where $\beta_{k}\in\mathbb{C}$ denotes the complex round-trip channel
coefficient of target $k\in\mathcal{K}$ depending on the associated
path loss and its radar cross section (RCS), $\boldsymbol{z}(n)\in\mathbb{C}^{N\times1}$
denotes the background noise at the BS receiver (including clutters
or interference), and $\boldsymbol{a}(\theta)$ denotes the steering
vector with an angle $\theta$, i.e.,\vspace{-2pt}
\begin{equation}
\begin{array}{c}
\boldsymbol{a}(\theta)=[1,e^{j2\pi\frac{d}{\lambda}\sin(\theta)},\ldots,e^{j2\pi(N-1)\frac{d}{\lambda}\sin(\theta)}]^{T}.\end{array}\label{eq:steer}
\end{equation}
In (\ref{eq:steer}), $\lambda$ denotes the carrier wavelength and
$d$ denotes the spacing between two adjacent antennas at the BS.

We consider the tracking task with the objective of estimating $\{\theta_{k}\}$
of the $K$ targets based on \eqref{eq:steer}, in which the BS \textit{a
priori} knows the number of targets $K$ and their initial angle estimation
$\{\hat{\theta_{k}}\}$ (e.g., based on the preceding detection) \cite{li2008mimo}.
To facilitate sensing design, for any sensing angle $\theta\in[-\frac{\pi}{2},\frac{\pi}{2}]$,
we define the beampattern gain $P(\theta)$ as the transmit signal
power distribution at that angle \cite{StoPETLiJ07,li2008mimo}, i.e.,\vspace{-2pt}
\begin{equation}
P(\theta)=\mathbb{E}\left(\big|\boldsymbol{a}^{H}(\theta)(\boldsymbol{s}_{1}+\boldsymbol{w}_{0}s_{0})\big|^{2}\right)=\boldsymbol{a}^{H}(\theta)(\boldsymbol{W}+\boldsymbol{S})\boldsymbol{a}(\theta).
\end{equation}
For any two sensing angle $\hat{\theta_{p}}$ and $\hat{\theta_{q}}$
($p,q\in\mathcal{K},p\neq q$), we define the cross-correlation pattern
as $\boldsymbol{a}^{H}(\hat{\theta_{p}})\boldsymbol{T}\boldsymbol{a}(\hat{\theta_{q}})$.

In order to optimize the sensing performance for estimation, the transmitted
beampattern should be designed to satisfy the following two requirements
\cite{StoPETLiJ07}. 
\begin{itemize}
\item First, we need to control the transmit beampattern $\{P(\theta)\}$
to mimic a desired transmit beampattern $\hat{P}(\theta)$ as follows
based on the estimated directions $\{\hat{\theta_{k}}\}$, \cite{LiuHuangNirJ20,SuLiuChrJ21,li2008mimo,LiStoJ07}\vspace{-2pt},
\begin{equation}
\hat{P}(\theta)=\begin{cases}
1 & \exists k\in\mathcal{K},|\theta-\hat{\theta_{k}}|<\frac{\Delta\theta}{2},\\
0 & \textrm{otherwise},
\end{cases}
\end{equation}
where $\Delta\theta$ denotes the width of desired beampattern at
each estimation angle. Towards this end, we adopt the beampattern
matching mean squared error (MSE), $B(\{\bar{\theta}_{m}\}_{m=1}^{M},\boldsymbol{S},\boldsymbol{W},\eta)$,
as the performance metric for quantifying the quality of beampattern
matching, which measures the difference between the actual transmit
beampattern $\{P(\theta)\}$ in the angular domain versus the desired
beampattern $\{\hat{P}(\theta)\}$, i.e.,\vspace{-2pt}
\begin{eqnarray}
B(\{\bar{\theta}_{m}\}_{m=1}^{M},\boldsymbol{S},\boldsymbol{W},\eta) & = & \frac{1}{M}\sum_{m=1}^{M}\big|\eta\hat{P}(\bar{\theta}_{m})-\boldsymbol{a}^{H}(\bar{\theta}_{m})(\boldsymbol{W}+\boldsymbol{S})\boldsymbol{a}(\bar{\theta}_{m})\big|^{2}.\label{eq:bme}
\end{eqnarray}
In (\ref{eq:bme}), $\{\bar{\theta}_{m}\}_{m=1}^{M}$ denote the $M$
sample angles over $[-\frac{\pi}{2},\frac{\pi}{2}]$ and $\eta$ is
a scaling factor to be optimized.
\item Next, we need to minimize the cross-correlation $|\boldsymbol{a}^{H}(\hat{\theta_{p}})\boldsymbol{T}\boldsymbol{a}(\hat{\theta_{q}})|$
at $\hat{\theta_{p}}$ and $\hat{\theta_{q}}$, for any $p,q\in\mathcal{K},p\neq q$
to resolve different target angles efficiently. Towards this end,
we use the mean squared value at different estimated target directions
$C(\{\hat{\theta_{k}}\},\boldsymbol{S},\boldsymbol{W})$ as the performance
metric \cite{li2008mimo,LiuHuangNirJ20}, i.e.,\vspace{-3pt}
\begin{equation}
C(\{\hat{\theta_{k}}\},\boldsymbol{S},\boldsymbol{W})=\frac{2}{K^{2}-K}\sum_{p=1}^{K}\sum_{q=p+1}^{K}\big|\boldsymbol{a}^{H}(\hat{\theta_{p}})(\boldsymbol{W}+\boldsymbol{S})\boldsymbol{a}(\hat{\theta_{q}})\big|^{2}.\label{eq:cross}
\end{equation}
\end{itemize}
To capture the importance of $B(\{\bar{\theta}_{m}\}_{m=1}^{M},\boldsymbol{S},\boldsymbol{W},\eta)$
in (\ref{eq:bme}) and $C(\{\hat{\theta_{k}}\},\boldsymbol{S},\boldsymbol{W})$
in (\ref{eq:cross}), we adopt the weighted sum of beampattern matching
MSE and cross-correlation pattern in the following as the sensing
performance metric to be minimized, similarly as that in \cite{StoPETLiJ07}:\vspace{-3pt}
\begin{eqnarray}
D(\boldsymbol{S},\boldsymbol{W},\eta) & = & \frac{1}{M}\sum_{m=1}^{M}\big|\eta\hat{P}(\bar{\theta}_{m})-\boldsymbol{a}^{H}(\bar{\theta}_{m})(\boldsymbol{W}+\boldsymbol{S})\boldsymbol{a}(\bar{\theta}_{m})\big|^{2}\nonumber \\
 &  & +\frac{2\omega_{C}}{K^{2}-K}\sum_{p=1}^{K}\sum_{q=p+1}^{K}\big|\boldsymbol{a}^{H}(\hat{\theta_{p}})(\boldsymbol{W}+\boldsymbol{S})\boldsymbol{a}(\hat{\theta_{q}})\big|^{2},\label{eq:beampattern design metric}
\end{eqnarray}
where $\omega_{C}\geq0$ is the weight of the cross-correlation term
that is pre-determined. As shown in \cite{StoPETLiJ07} for MIMO radar
sensing, by designing $\boldsymbol{W}$ and $\boldsymbol{S}$ to properly
optimize the weighted sum in (\ref{eq:beampattern design metric}),
the multi-target estimation performance can be efficiently optimized.\footnote{The multi-target estimation performance in secure ISAC of our interest
will be shown in Section V using the practical Capon estimation algorithm
\cite{xu2008target,StoPETLiJ07,li2008mimo}.}

\subsection{Problem Formulation}

Our objective is to minimize the weighted sum of beampattern matching
MSE and cross-correlation patterns $D(\boldsymbol{S},\boldsymbol{W},\eta)$
in (\ref{eq:beampattern design metric}), by jointly optimizing the
transmit information covariance matrix $\boldsymbol{W}$ and the sensing/AN
covariance matrix $\boldsymbol{S}$, subject to the secure communication
requirements under two scenarios with different CSI error models of
eavesdroppers.

First, we consider the scenario with the bounded CSI errors of eavesdroppers.
In this scenario, the worst-case secrecy rate in \eqref{eqn:outage:probability}
is adopted as the performance metric. The worst-case secrecy rate
constrained beampattern optimization problem is formulated as\vspace{-5pt}\begin{subequations}
\begin{eqnarray}
\textrm{(P1)}: & \underset{\boldsymbol{W},\boldsymbol{S},\eta}{\min} & D(\boldsymbol{W},\boldsymbol{S},\eta)\nonumber \\
 & \textrm{s.t.} & R_{\mathrm{w}}(\boldsymbol{W},\boldsymbol{S})\geq R_{0},\label{eq:worst-case secrecy rate constraint}\\
 &  & \mathrm{Tr}(\boldsymbol{E}_{i}(\boldsymbol{W}+\boldsymbol{S}))=Q,\forall i\in\{1,\dots,N\},\label{eq:uniform element power}\\
 &  & \boldsymbol{W}\succeq\boldsymbol{0},\boldsymbol{S}\succeq\boldsymbol{0},\label{eq:semidefinite constraint}\\
 &  & \textrm{rank}(\boldsymbol{W})\le1,\label{eq:rank constraint}
\end{eqnarray}
\end{subequations}where $R_{0}$ denotes the constant worst-case
secrecy rate threshold required at the CU. Notice that problem (P1)
is non-convex due to the non-convex worst-case secrecy rate constraint
(\ref{eq:worst-case secrecy rate constraint}) and the rank constraint
(\ref{eq:rank constraint}), and thus is difficult to solve. Despite
the difficulty, we will optimally solve problem (P1) in Section III.

Next, we consider the Gaussian CSI errors of eavesdroppers. In this
scenario, we need to ensure the secrecy outage probability $\mathrm{Pr}(R(\boldsymbol{W},\boldsymbol{S})<R_{0})$
in \eqref{eqn:worst_case:secrecy} not to exceed a certain threshold
$\rho$ to satisfy the secure communication performance. The secrecy
outage-constrained beampattern optimization problem is formulated
as\vspace{-5pt} 
\begin{subequations}
\begin{eqnarray}
\textrm{(P2):} & \underset{\boldsymbol{W},\boldsymbol{S},\eta}{\min} & D(\boldsymbol{W},\boldsymbol{S},\eta)\nonumber \\
 & \textrm{s.t.} & \textrm{Pr}(R(\boldsymbol{W},\boldsymbol{S})<R_{0})\leq\rho,\label{eq:Secrecy outage}\\
 &  & \mathrm{Tr}(\boldsymbol{E}_{i}(\boldsymbol{W}+\boldsymbol{S}))=Q,\forall i\in\{1,\dots,N\},\label{eq:Power constraint}\\
 &  & \boldsymbol{W}\succeq\boldsymbol{0},\boldsymbol{S}\succeq\boldsymbol{0},\label{eq:Semi definite condtraint}\\
 &  & \textrm{rank}(\boldsymbol{W})\le1,\label{eq:rank constraint-1}
\end{eqnarray}
\end{subequations}
where $0<\rho<0.5$ is the given outage threshold. In problem (P2),
the rank constraint in (\ref{eq:rank constraint-1}) is non-convex
and the secrecy outage constraint in (\ref{eq:Secrecy outage}) even
does not admit a closed-form expression. Therefore, problem (P2) is
more difficult to solve than (P1). We will propose an efficient algorithm
to solve problem (P2) in Section IV.
\begin{rem}
\label{thm:It-is-worth}It is worth emphasizing that for the special
case with perfect CSI of eavesdroppers, the simplified secrecy rate-constrained
beampattern optimization problem (i.e., (P1) with $\boldsymbol{e}_{k}=\boldsymbol{0}$
or $\varepsilon_{k}=0$, $\forall k\in\mathcal{K}_{\mathrm{E}}$ and
(P2) with $\boldsymbol{C}_{k}=\boldsymbol{0}$, $\forall k\in\mathcal{K}_{\mathrm{E}}$
and $\rho=0$) is novel and has not been addressed in the literature
yet, e.g., \cite{SuLiuChrJ21,DelAnaJ18}. As will be shown in Section
III, this simplified problem can be solved optimally based on the
optimal solution to (P1). Therefore, we will not address this subcase
separately. Instead, we will present numerical results for it in Section
V-A to obtain more insights. 
\end{rem}

\section{Optimal Joint Beamforming Solution to Problem (P1)}

This section presents the optimal joint transmit beamforming solution
to problem (P1). Towards this end, we first deal with the rank constraint
in (\ref{eq:rank constraint}) based on the idea of SDR. By dropping
the rank constraint in (\ref{eq:rank constraint}), we obtain a relaxed
version of (P1) as\vspace{-5pt}
\begin{eqnarray*}
\textrm{(P1.1)}: & \underset{\boldsymbol{W},\boldsymbol{S},\eta}{\min} & D(\boldsymbol{W},\boldsymbol{S},\eta)\\
 & \textrm{s.t.} & \textrm{(\textrm{\ref{eq:worst-case secrecy rate constraint}}), (\ref{eq:uniform element power}),\textrm{ and }(\textrm{\ref{eq:semidefinite constraint}})}.
\end{eqnarray*}
In the following subsections, we first obtain the optimal solution
to problem (P1.1) with general-rank information covariance $\boldsymbol{W}$.
Then, we propose a construction method to find an optimal rank-one
solution of $\boldsymbol{W}$ for (P1).

\subsection{Optimal Solution to (P1.1)}

In the following, we focus on solving problem (P1.1) with the non-convex
worst-case secrecy rate constraint (\ref{eq:worst-case secrecy rate constraint}).
First, we introduce an auxiliary optimization variable $\gamma$ to
represent the \ac{sinr} threshold at the \ac{cu} such that the
worst-case secrecy rate constraint (\ref{eq:worst-case secrecy rate constraint})
is equivalently reformulated as 
\begin{equation}
\frac{\boldsymbol{g}^{H}\boldsymbol{W}\boldsymbol{g}}{\boldsymbol{g}^{H}\boldsymbol{S}\boldsymbol{g}+\sigma_{0}^{2}}\geq\gamma,\underset{\|\boldsymbol{e}_{k}\|\leq\varepsilon_{k}}{\max}\frac{(\hat{\boldsymbol{h}}_{k}+\boldsymbol{e}_{k})^{H}\boldsymbol{W}(\hat{\boldsymbol{h}}_{k}+\boldsymbol{e}_{k})}{(\hat{\boldsymbol{h}}_{k}+\boldsymbol{e}_{k})^{H}\boldsymbol{S}(\hat{\boldsymbol{h}}_{k}+\boldsymbol{e}_{k})^{H}+\sigma_{k}^{2}}\leq\varphi(\gamma),\forall k\in\mathcal{K}_{\mathrm{E}},
\end{equation}
where $\varphi(\gamma)=2^{-R_{0}}(1+\gamma)-1$, and $(\cdot)^{+}$
in the worst-case secrecy rate formula in (\ref{eq:worst-case secrecy rate constraint})
is omitted without loss of optimality. Notice that in order for $\gamma$
to be feasible to (P1.1), it must hold that $\gamma\in\Gamma\triangleq\{\gamma|2^{R_{0}}-1\leq\gamma\leq NQ\|\boldsymbol{g}\|^{2}/\sigma_{0}^{2}\}$.
Therefore, 
problem (P1.1) is equivalently expressed as\vspace{-5pt}
\begin{subequations}
\begin{eqnarray}
\textrm{(P1.2)}: & \underset{\boldsymbol{W},\boldsymbol{S},\eta,\mathscr{\gamma\in\varGamma}}{\min} & D(\boldsymbol{W},\boldsymbol{S},\eta)\nonumber \\
 & \textrm{s.t.} & \boldsymbol{g}^{H}\boldsymbol{W}\boldsymbol{g}\geq\gamma(\boldsymbol{g}^{H}\boldsymbol{S}\boldsymbol{g}+\sigma_{0}^{2}),\label{eq:CU constraint}\\
 &  & \boldsymbol{e}_{k}^{H}\big(\boldsymbol{W}-\varphi(\gamma)\boldsymbol{S}\big)\boldsymbol{e}_{k}+2\textrm{Re}\Big(\boldsymbol{e}_{k}^{H}\big(\boldsymbol{W}-\varphi(\gamma)\boldsymbol{S}\big)\hat{\boldsymbol{h}}_{k}\Big)\nonumber \\
 &  & +\hat{\boldsymbol{h}}_{k}^{H}\big(\boldsymbol{W}-\varphi(\gamma)\boldsymbol{S}\big)\hat{\boldsymbol{h}}_{k}-\varphi(\gamma)\sigma_{k}^{2}\leq0,\forall\|\boldsymbol{e}_{k}\|\leq\varepsilon_{k},\forall k\in\mathcal{K}_{\mathrm{E}},\label{eq:target constraint}\\
 &  & \textrm{(\ref{eq:uniform element power}) and (\textrm{\ref{eq:semidefinite constraint}})}.\nonumber 
\end{eqnarray}
\end{subequations}
\setlength{\belowdisplayskip}{3pt}Then, we introduce the S-procedure
in the following lemma \cite{WangAntJ14} to handle constraint (\ref{eq:target constraint}). 
\begin{lem}
\label{lem:-S-procedure:-Let-1}Suppose that $f_{i}(\boldsymbol{v})=\boldsymbol{v}^{H}\boldsymbol{M}_{i}\boldsymbol{v}+2\mathcal{R}e\left\{ \boldsymbol{v}^{H}\boldsymbol{b}_{i}\right\} +n_{i},\boldsymbol{M}_{i}\in\mathbb{S}^{N},\boldsymbol{b}_{i}\in\mathbb{C}^{N\times1},i=1,2$,
where there exists a point $\hat{\boldsymbol{v}}$ such that $f_{1}(\hat{\boldsymbol{v}})<0$.
Then, it follows that $f_{1}(\boldsymbol{v})\leq0\Longrightarrow f_{2}(\boldsymbol{v})\leq0$
if and only if there exists $\lambda\geq0$ such that 
\begin{equation}
\lambda\left[\begin{array}{cc}
\boldsymbol{M}_{1} & \boldsymbol{b}_{1}\\
\boldsymbol{b}_{1}^{H} & n_{1}
\end{array}\right]-\left[\begin{array}{cc}
\boldsymbol{M}_{2} & \boldsymbol{b}_{2}\\
\boldsymbol{b}_{2}^{H} & n_{2}
\end{array}\right]\succeq\boldsymbol{0}.
\end{equation}
\end{lem}
Notice that constraint (\ref{eq:target constraint}) is equivalent
to 
\begin{eqnarray}
\boldsymbol{e}_{k}^{H}\boldsymbol{e}_{k}-\varepsilon_{k}^{2}\leq0 & \Rightarrow & \boldsymbol{e}_{k}^{H}\big(\boldsymbol{W}-\varphi(\gamma)\boldsymbol{S}\big)\boldsymbol{e}_{k}+2\textrm{Re}\Big(\boldsymbol{e}_{k}^{H}\big(\boldsymbol{W}-\varphi(\gamma)\boldsymbol{S}\big)\hat{\boldsymbol{h}}_{k}\Big)\nonumber \\
 &  & +\hat{\boldsymbol{h}}_{k}^{H}\big(\boldsymbol{W}-\varphi(\gamma)\boldsymbol{S}\big)\hat{\boldsymbol{h}}_{k}-\varphi(\gamma)\sigma_{k}^{2}\leq0,\forall k\in\mathcal{K}_{\mathrm{E}}.\label{eq:targae constraint 2}
\end{eqnarray}
Therefore, based on Lemma \ref{lem:-S-procedure:-Let-1}, constraint
(\ref{eq:targae constraint 2}) or equivalent (\ref{eq:target constraint})
is equivalently rewritten as 
\begin{eqnarray}
\hspace*{-0.3cm}\left[\begin{array}{cc}
\lambda_{k}\boldsymbol{I}-\big(\boldsymbol{W}-\varphi(\gamma)\boldsymbol{S}\big) & -\big(\boldsymbol{W}-\varphi(\gamma)\boldsymbol{S}\big)\hat{\boldsymbol{h}}_{k}\\
-\hat{\boldsymbol{h}}_{k}^{H}\big(\boldsymbol{W}-\varphi(\gamma)\boldsymbol{S}\big) & -\lambda_{k}\varepsilon_{k}^{2}-\hat{\boldsymbol{h}}_{k}^{H}\big(\boldsymbol{W}-\varphi(\gamma)\boldsymbol{S}\big)\hat{\boldsymbol{h}}_{k}+\varphi(\gamma)\sigma_{k}^{2}
\end{array}\right]\hspace{-0.1cm}\succeq\hspace{-0.1cm}\boldsymbol{0},\lambda_{k}\geq0,\hspace*{-0.1cm}\forall k\in\mathcal{K}_{\mathrm{E}},\label{eq:LMI constriant}
\end{eqnarray}
where $\{\lambda_{k}\}$ are the new auxiliary optimization variables.
Hence, problem (P1.2) can be equivalently reformulated as\setlength{\abovedisplayskip}{3pt}
\begin{subequations}
\begin{eqnarray*}
\textrm{(P1.3)}: & \underset{\boldsymbol{W},\boldsymbol{S},\eta,\gamma\in\varGamma,\{\lambda_{k}\}}{\min} & D(\boldsymbol{W},\boldsymbol{S},\eta)\\
 & \textrm{s.t.} & \textrm{(\ref{eq:uniform element power}), (\ref{eq:semidefinite constraint}), (\ref{eq:CU constraint}), and (\ref{eq:LMI constriant}).}
\end{eqnarray*}
\end{subequations}
\setlength{\belowdisplayskip}{3pt}Problem (P1.3) is still non-convex
due to the non-convex constraints in (\ref{eq:CU constraint}) and
(\ref{eq:LMI constriant}). To facilitate the development of solution
to (P1.3), we consider the optimization of $\boldsymbol{W},\boldsymbol{S},\eta$,
and $\{\lambda_{k}\}$ in problem (P1.3) under any given $\gamma\in\varGamma$
as\setlength{\abovedisplayskip}{3pt}
\begin{subequations}
\begin{eqnarray*}
\textrm{(P1.4)}: & \underset{\boldsymbol{W},\boldsymbol{S},\eta,\{\lambda_{k}\}}{\min} & D(\boldsymbol{W},\boldsymbol{S},\eta)\\
 & \textrm{s.t.} & \textrm{(\ref{eq:uniform element power}), (\textrm{\ref{eq:semidefinite constraint}}), (\ref{eq:CU constraint}), and (\ref{eq:LMI constriant}).}
\end{eqnarray*}
\end{subequations}
\setlength{\belowdisplayskip}{3pt}It is observed that with given
$\gamma$, (\ref{eq:CU constraint}) is an affine constraint and (\ref{eq:LMI constriant})
is a \ac{lmi} constraint, both of which are convex. Therefore,
problem (P1.4) is a convex \ac{qsdp} problem that can be optimally
solvable via off-the-shelf convex programming numerical solvers such
as CVX \cite{cvx}.

Let $f(\gamma)$ denote the optimal objective value achieved by problem
(P1.4) with given $\gamma$. Then we solve problem (P1.3) by first
solving problem (P1.4) under any given $\gamma$ and then searching
over $\gamma\in\varGamma$ via a \ac{1d} search in the following
problem:
\begin{eqnarray*}
(\textrm{P1.5}): & \underset{_{\gamma\in\varGamma}}{\min} & f(\gamma).
\end{eqnarray*}
Therefore, by optimally solving problem (P1.4) together with the 1D
search for solving problem (P1.5), the optimal solution to problem
(P1.3) is obtained, which is denoted as $\tilde{\boldsymbol{W}}^{*}$,
$\tilde{\boldsymbol{S}}^{*}$, $\tilde{\gamma}^{*}$, $\tilde{\eta}^{*}$,
and $\{\tilde{\lambda}^{*}\}$. Accordingly, $\tilde{\boldsymbol{W}}^{*}$,
$\tilde{\boldsymbol{S}}^{*}$, and $\tilde{\eta}^{*}$ are the optimal
solution to problem (P1.1). Notice that $\textrm{rank}(\tilde{\boldsymbol{W}}^{*})\leq1$
may not hold in general. Therefore, the obtained $\tilde{\boldsymbol{W}}^{*}$,
$\tilde{\boldsymbol{S}}^{*}$, $\textrm{ and }\tilde{\eta}^{*}$ from
the numerical solvers may not be always feasible for problem (P1).

\subsection{Construction of Rank-one Solution of $\boldsymbol{W}$ for (P1)}

To address the aforementioned issue, we propose a novel approach to
construct a rank-one solution of $\boldsymbol{W}$ to problem (P1)
if $\textrm{rank}(\tilde{\boldsymbol{W}}^{*})>1$. In the following
proposition, we construct an optimal solution to problem (P1.3) satisfying
the rank-one constraint in (\ref{eq:rank constraint}). 
\begin{prop}
\textup{\label{prop:1}Based on the obtained optimal solution $\tilde{\boldsymbol{W}}^{*}$,
$\tilde{\boldsymbol{S}}^{*}$, $\tilde{\gamma}^{*}$, $\tilde{\eta}^{*}$,
and $\{\tilde{\lambda_{k}}^{*}\}$ to problem (P1.3), we can always
construct an equivalent solution of $\boldsymbol{W}^{*}$, $\boldsymbol{S}^{*}$,
$\gamma^{*}$, $\eta^{*}$, and $\{\lambda_{k}^{*}\}$ in the following,
which is also optimal to problem (P1.3) with $\textrm{rank}(\boldsymbol{W}^{*})=1$.
\begin{subequations}
\begin{eqnarray}
 &  & \boldsymbol{W}^{*}=\frac{\tilde{\boldsymbol{W}}^{*}\boldsymbol{g}\boldsymbol{g}^{H}\tilde{\boldsymbol{W}}^{*}}{\boldsymbol{g}^{H}\tilde{\boldsymbol{W}}^{*}\boldsymbol{g}},\label{eq:construction W}\\
 &  & \boldsymbol{S}^{*}=\tilde{\boldsymbol{W}}^{*}+\tilde{\boldsymbol{S}}^{*}-\boldsymbol{W}^{*},\label{eq:construction S}\\
 &  & \eta^{*}=\tilde{\eta}^{*},\gamma^{*}=\tilde{\gamma}^{*},\lambda_{k}^{*}=\tilde{\lambda_{k}}^{*},\forall k\in\mathcal{K}_{\mathrm{E}}.\label{eq:construction others}
\end{eqnarray}
\end{subequations}
} 
\end{prop}
\begin{proof} Please refer to Appendix A. \end{proof} Proposition
\ref{prop:1} shows the existence of the optimal rank-one information
covariance matrix to problem (P1.3) and thus (P1), which validates
that the adopted SDR of problem (P1) is tight. Therefore, the constructed
solution $\boldsymbol{W}^{*}$, $\boldsymbol{S}^{*}$, and $\eta^{*}$
is optimal to problem (P1.1) with $\textrm{rank}(\boldsymbol{W}^{*})=1$
thus is also optimal to problem (P1).

\section{Proposed Joint Beamforming Solution to Problem (P2)}

In this section, we propose an efficient solution to the secrecy outage
constrained sensing beampattern optimization problem (P2). Notice
that the secrecy outage constraint in (\ref{eq:Secrecy outage}) is
equivalent to 
\begin{equation}
\textrm{Pr}(R(\boldsymbol{W},\boldsymbol{S})\geq R_{0})\geq1-\rho.\label{eq:secrecy non-outage}
\end{equation}
By dropping the rank constraint in (\ref{eq:rank constraint-1}) similarly
as for problem (P1) and replacing (\ref{eq:Secrecy outage}) as (\ref{eq:secrecy non-outage}),
we relax problem (P2) as
\begin{subequations}
\begin{eqnarray*}
\textrm{(P2.1)}: & \underset{\boldsymbol{W},\boldsymbol{S},\eta}{\min} & D(\boldsymbol{W},\boldsymbol{S},\eta)\\
 & \textrm{s.t.} & \textrm{ (\ref{eq:Power constraint}), (\textrm{\ref{eq:Semi definite condtraint}})\textrm{, and }(\textrm{\ref{eq:secrecy non-outage}})}.
\end{eqnarray*}
\end{subequations}
By introducing an auxiliary optimization variable $\gamma$, the secrecy
outage constraint in (\ref{eq:secrecy non-outage}) is equivalently
reformulated as 
\begin{eqnarray}
 &  & \frac{\boldsymbol{g}^{H}\boldsymbol{W}\boldsymbol{g}}{\boldsymbol{g}^{H}\boldsymbol{S}\boldsymbol{g}+\sigma_{0}^{2}}\geq\gamma,\label{eq:SINR for CU}\\
 &  & \textrm{Pr}\Big(\log_{2}(1+\gamma)-\underset{k\in\mathcal{K}_{\mathrm{E}}}{\max}\log_{2}\big(1+\frac{\boldsymbol{h}_{k}^{H}\boldsymbol{W}\boldsymbol{h}_{k}}{\boldsymbol{h}_{k}^{H}\boldsymbol{S}\boldsymbol{h}_{k}+\sigma_{k}^{2}}\big)\geq R_{0}\Big)\geq1-\rho.\label{eq:new secrecy outage}
\end{eqnarray}
We denote the feasible region of $\gamma$ as $\varGamma=\{\gamma|2^{R_{0}}-1\leq\gamma\leq NQ\|\boldsymbol{g}\|^{2}/\sigma_{0}^{2}\}$.
Thus, problem (P2.1) is equivalently rewritten as
\begin{subequations}
\begin{eqnarray*}
\textrm{(P2.2)}: & \underset{\boldsymbol{W},\boldsymbol{S},\eta,\gamma\in\varGamma}{\min} & D(\boldsymbol{W},\boldsymbol{S},\eta)\\
 & \textrm{s.t.} & \textrm{(\ref{eq:Power constraint}), (\textrm{\ref{eq:Semi definite condtraint}}), (\textrm{\ref{eq:SINR for CU}}), \textrm{and }(\ref{eq:new secrecy outage})}.
\end{eqnarray*}
\end{subequations}

To solve problem (P2.2), we consider the optimization of $\boldsymbol{W},\boldsymbol{S},\textrm{ and }\eta$
under given $\gamma$, which is given by 
\begin{eqnarray*}
\textrm{(P2.3):} & \underset{\boldsymbol{W},\boldsymbol{S},\eta}{\min} & D(\boldsymbol{W},\boldsymbol{S},\eta)\\
 & \textrm{s.t.} & \textrm{(\ref{eq:Power constraint}), (\textrm{\ref{eq:Semi definite condtraint}}), (\textrm{\ref{eq:SINR for CU}}), \textrm{and }(\ref{eq:new secrecy outage})}.
\end{eqnarray*}
We denote the obtained objective value for problem (P2.3) as $f_{2}(\gamma)$.
As a result, we solve problem (P2.2) equivalently by first solving
(P2.3) under any given $\gamma\in\mathscr{\varGamma}$, and then find
the optimal $\gamma$ that minimizes $f_{2}(\gamma)$ over $\mathscr{\varGamma}$
via a 1D search.

Now, we only need to focus on solving problem (P2.3) with given $\gamma\in\varGamma$.
In the following, we first propose a safe decomposition for constraint
(\ref{eq:new secrecy outage}), then exploit the convex restriction
technique of \ac{bti} to establish restricted surrogate convex
problems for solving (P2.3), and finally construct a high-quality
solution to the original problem (P2) based on the solution to (P2.3).

\subsection{Safe Decomposition for Constraint\textmd{\normalsize{} $\text{(\ref{eq:new secrecy outage})}$}}

To solve problem (P2.3), we first deal with the secrecy outage constraint
(\ref{eq:new secrecy outage}). Recall that $\boldsymbol{h}_{k}$'s
are the channel vectors of different eavesdroppers that are assumed
to be independent. Therefore, we have 
\begin{equation}
\begin{array}{cl}
 & \textrm{Pr}\Big(\log_{2}(1+\gamma)-\underset{k\in\mathcal{K}_{\mathrm{E}}}{\max}\log_{2}\big(1+\frac{\boldsymbol{h}_{k}^{H}\boldsymbol{W}\boldsymbol{h}_{k}}{\boldsymbol{h}_{k}^{H}\boldsymbol{S}\boldsymbol{h}_{k}+\sigma_{k}^{2}}\big)\geq R_{0}\Big)\\
= & \prod_{k\in\mathcal{K}_{\mathrm{E}}}\textrm{Pr}\Big(\log_{2}(1+\gamma)-\log_{2}\big(1+\frac{\boldsymbol{h}_{k}^{H}\boldsymbol{W}\boldsymbol{h}_{k}}{\boldsymbol{h}_{k}^{H}\boldsymbol{S}\boldsymbol{h}_{k}+\sigma_{k}^{2}}\big)\geq R_{0}\Big).
\end{array}\label{eq:independent hks}
\end{equation}
Based on (\ref{eq:independent hks}), it follows that constraint (\ref{eq:new secrecy outage})
is equivalent to 
\begin{eqnarray}
 & \prod_{k\in\mathcal{K}_{\mathrm{E}}}\textrm{Pr}\Big(\frac{\boldsymbol{h}_{k}^{H}\boldsymbol{W}\boldsymbol{h}_{k}}{\boldsymbol{h}_{k}^{H}\boldsymbol{S}\boldsymbol{h}_{k}+\sigma_{k}^{2}}\leq\varphi(\gamma)\Big)\geq1-\rho,\label{eq:Product of SINR at targets}
\end{eqnarray}
where $\varphi(\gamma)=2^{-R_{0}}(\gamma+1)-1$. However, it is still
difficult to express constraint (\ref{eq:Product of SINR at targets})
into a tractable form. To resolve this issue, we replace (\ref{eq:Product of SINR at targets})
by the following $K_{\mathrm{E}}$ restricted constraints\footnote{Note that the equivalence between (\ref{eq:new secrecy outage}) and
(\ref{eq:Product of SINR at targets}) holds based on the assumption
that $\boldsymbol{h}_{k}$'s are independent and the restriction from
(\ref{eq:Product of SINR at targets}) to (\ref{eq:probability approxiamation})
does not depend on the independence of $\{\boldsymbol{h}_{k}\}$,
since the feasible set of $\boldsymbol{W}$ and $\boldsymbol{S}$
characterized by (\ref{eq:probability approxiamation}) is always
a subset of that by (\ref{eq:Product of SINR at targets}). } \cite{li2013safe} 
\begin{equation}
\textrm{Pr}\Big(\frac{\boldsymbol{h}_{k}^{H}\boldsymbol{W}\boldsymbol{h}_{k}}{\boldsymbol{h}_{k}^{H}\boldsymbol{S}\boldsymbol{h}_{k}+\sigma_{k}^{2}}\leq\varphi(\gamma)\Big)\geq1-\bar{\rho},\forall k\in\mathcal{K}_{\mathrm{E}},\label{eq:probability approxiamation}
\end{equation}
with $\bar{\rho}=1-(1-\rho)^{\frac{1}{K_{\mathrm{E}}}}$, which ensure
that the probability of each untrusted target's SINR being less than
threshold $\varphi(\gamma)$ is less than $\bar{\rho}$ .

Now, we deal with (\ref{eq:probability approxiamation}). Recall that
$\boldsymbol{h}_{k}=\hat{\boldsymbol{h}_{k}}+\boldsymbol{C}_{k}^{\frac{1}{2}}\boldsymbol{u}_{k},\forall k\in\mathcal{K}_{\mathrm{E}},$
where $\boldsymbol{u}_{k}\sim\!\mathcal{CN}(\boldsymbol{0},\boldsymbol{I})$.
Accordingly, it follows that 
\begin{eqnarray}
\frac{\boldsymbol{h}_{k}^{H}\boldsymbol{W}\boldsymbol{h}_{k}}{\boldsymbol{h}_{k}^{H}\boldsymbol{S}\boldsymbol{h}_{k}+\sigma_{k}^{2}}\leq\varphi(\gamma)\Longleftrightarrow\boldsymbol{u}_{k}^{H}\boldsymbol{A}_{k}\boldsymbol{u}_{k}+2\mathrm{Re}\big\{\boldsymbol{u}_{k}^{H}\boldsymbol{q}_{k}\big\}+c_{k} & \geq & 0,k\in\mathcal{K}_{\mathrm{E}},\label{eq:benin}
\end{eqnarray}
where\begin{subequations} 
\begin{eqnarray}
\boldsymbol{A}_{k} & = & \boldsymbol{C}_{k}^{\frac{1}{2}}\big(\varphi(\gamma)\boldsymbol{S}-\boldsymbol{W}\big)\boldsymbol{C}_{k}^{\frac{1}{2}},\\
\boldsymbol{q}_{k} & = & \boldsymbol{C}_{k}^{\frac{1}{2}}\big(\varphi(\gamma)\boldsymbol{S}-\boldsymbol{W}\big)\hat{\boldsymbol{h}_{k}},\\
c_{k} & = & \hat{\boldsymbol{h}_{k}}^{H}\big(\varphi(\gamma)\boldsymbol{S}-\boldsymbol{W}\big)\hat{\boldsymbol{h}_{k}}+\sigma_{k}^{2}\varphi(\gamma).
\end{eqnarray}
\end{subequations} Based on (\ref{eq:benin}), the constraints in
(\ref{eq:probability approxiamation}) are equivalent to 
\begin{equation}
\textrm{Pr}\Big(\boldsymbol{u}_{k}^{H}\boldsymbol{A}_{k}\boldsymbol{u}_{k}+2\mathrm{Re}\big\{\boldsymbol{u}_{k}^{H}\boldsymbol{q}_{k}\big\}+c_{k}\geq0\Big)\geq1-\bar{\rho},\forall k\in\mathcal{K}_{\mathrm{E}}.\label{eq:restricted outage}
\end{equation}

By substituting (\ref{eq:new secrecy outage}) as (\ref{eq:restricted outage}),
a restricted version of (P2.3) is obtained as 
\begin{eqnarray*}
\textrm{(P2.4):} & \underset{\boldsymbol{W},\boldsymbol{S},\eta}{\min} & D(\boldsymbol{W},\boldsymbol{S},\eta)\\
 & \textrm{s.t.} & \textrm{(\ref{eq:Power constraint}), (\textrm{\ref{eq:Semi definite condtraint}}), (\textrm{\ref{eq:SINR for CU}}), and \text{(\ref{eq:restricted outage})}}.
\end{eqnarray*}
For problem (P2.4), although there is still no closed-form expression
for (\ref{eq:restricted outage}), there have been some well-established
approaches such as the \ac{bti} \cite{WangAntJ14} to transform
(\ref{eq:restricted outage}) into restricted convex constraints.
In the following subsection, we solve problem (P2.4) via the convex
restriction technique based on the \ac{bti}.

\subsection{Convex Restriction of Problem (P2.4) via BTI }

This subsection solves problem (P2.4) by reformulating it into a restricted
problem based on the \ac{bti} in the following lemma \cite{MaHongJ14}. 
\begin{lem}
\label{lem:-(Bernstein-type-inequality)}For any $\boldsymbol{A}\in\mathbb{S}^{n},\boldsymbol{q}\in\mathbb{C}^{n},c\in\mathbb{R},\boldsymbol{v}\sim\mathcal{CN}(\boldsymbol{0},\boldsymbol{I})$,
and $\rho\in(0,1]$, if there exist $a$ and $b$, such that 
\begin{eqnarray}
 &  & \mathrm{Tr}(\boldsymbol{A})-\sqrt{-2\ln(\rho)}\cdot a+\ln(\rho)\cdot b+c\geq0,\label{eq:bti1}\\
 &  & \left\Vert \left[\begin{array}{c}
\sqrt{2}\boldsymbol{q}\\
\mathrm{vec}(\boldsymbol{A})
\end{array}\right]\right\Vert \leq a,\label{eq:bti2}\\
 &  & b\boldsymbol{I}+\boldsymbol{A}\succeq\mathbf{0},\quad b\geq0,\label{eq:bti3}
\end{eqnarray}
the following implication holds true 
\begin{equation}
\mathrm{Pr}\big(\boldsymbol{v}^{H}\boldsymbol{A}\boldsymbol{v}+2\mathrm{Re}\left\{ \boldsymbol{v}^{H}\boldsymbol{q}\right\} +c\geq0\big)\geq1-\rho.
\end{equation}
\end{lem}
Based on Lemma \ref{lem:-(Bernstein-type-inequality)}, it follows
that (\ref{eq:restricted outage}) can be restricted as the following
set of inequalities.\begin{subequations} 
\begin{eqnarray}
 &  & \mathrm{Tr}(\boldsymbol{A}_{k})-\sqrt{-2\ln(\bar{\rho})}\cdot a_{k}+\ln(\bar{\rho})\cdot b_{k}+c_{k}\geq0,\forall k\in\mathcal{K}_{\mathrm{E}},\label{eq:bti1-1}\\
 &  & \left\Vert \left[\begin{array}{c}
\sqrt{2}\boldsymbol{q}_{k}\\
\mathrm{vec}(\boldsymbol{A}_{k})
\end{array}\right]\right\Vert \leq a_{k},\forall k\in\mathcal{K}_{\mathrm{E}},\label{eq:bti2-1}\\
 &  & b_{k}\boldsymbol{I}+\boldsymbol{A}_{k}\succeq\mathbf{0},\quad b_{k}\geq0,\forall k\in\mathcal{K}_{\mathrm{E}}.\label{eq:bti3-1}
\end{eqnarray}
\end{subequations} 

In other words, if (\ref{eq:bti1-1}), (\ref{eq:bti2-1}), and (\ref{eq:bti3-1})
hold, then (\ref{eq:restricted outage}) is satisfied. By substituting
(\ref{eq:restricted outage}) as (\ref{eq:bti1-1}), (\ref{eq:bti2-1}),
and (\ref{eq:bti3-1}), we obtain a restricted version of problem
(P2.4) as 
\begin{eqnarray*}
\textrm{(P2.5):} & \underset{\boldsymbol{W},\boldsymbol{S},\eta,\{a_{k}\},\{b_{k}\}}{\min} & D(\boldsymbol{W},\boldsymbol{S},\eta)\\
 & \textrm{s.t.} & \textrm{(\ref{eq:Power constraint}), (\textrm{\ref{eq:Semi definite condtraint}}), (\textrm{\ref{eq:SINR for CU}}), (\ref{eq:bti1-1}), (\ref{eq:bti2-1}), and (\ref{eq:bti3-1}) }.
\end{eqnarray*}
Here, the restricted problem (P2.5) is convex and thus can be optimally
solved via CVX, for which the obtained optimal solution is a high-quality
solution to problem (P2.4).

So far, by solving restricted problem (P2.5), problem (P2.4) has been
efficiently solved with a given $\gamma\in\Gamma$. Based on this
together with the \ac{1d} search over $\gamma\in\Gamma$, an efficient
solution to problem (P2.1) is also obtained, which is denoted by $\tilde{\boldsymbol{W}}^{\star},\tilde{\boldsymbol{S}}^{\star},\textrm{ and }\tilde{\eta}^{\star}$.
However, $\tilde{\boldsymbol{W}}^{\star}$ may not satisfy $\textrm{rank}(\tilde{\boldsymbol{W}}^{\star})\le1$
in general and thus the obtained solution may not be feasible to the
original problem (P2).

\subsection{Rank-one Construction of $\boldsymbol{W}$ for Problem (P2)}

In this subsection, we construct a rank-one solution to problem (P2),
as stated in the following proposition. 
\begin{prop}
\textup{\label{prop:rank-one construction for prob error}Based on
the obtained solution }$\tilde{\boldsymbol{W}}^{\star},\tilde{\boldsymbol{S}}^{\star},$\textup{
and $\tilde{\eta}^{\star}$ to problem (P2.1), we can always construct
an equivalent solution $\boldsymbol{W}^{\star},\boldsymbol{S}^{\star},\textrm{ and }\eta^{\star}$,
such that the same objective value in (P2.1) is achieved with rank($\boldsymbol{W}^{\star}$)
= 1. 
\begin{subequations}
\begin{eqnarray}
 &  & \boldsymbol{W}^{\star}=\frac{\tilde{\boldsymbol{W}}^{\star}\boldsymbol{g}\boldsymbol{g}^{H}\tilde{\boldsymbol{W}}^{\star}}{\boldsymbol{g}^{H}\tilde{\boldsymbol{W}}^{\star}\boldsymbol{g}},\label{eq:construction W new}\\
 &  & \boldsymbol{S}^{\star}=\tilde{\boldsymbol{W}}^{\star}+\tilde{\boldsymbol{S}}^{\star}-\boldsymbol{W}^{\star},\label{eq:construction S new}\\
 &  & \eta^{\star}=\tilde{\eta}^{\star}.\label{eq:new eta}
\end{eqnarray}
\end{subequations}
As a result, the constructed solution of $\boldsymbol{W}^{\star},\boldsymbol{S}^{\star},\textrm{ and }\eta^{\star}$
is a feasible high-quality solution to problem (P2).} 
\end{prop}
\begin{proof} Please refer to Appendix B. \end{proof}
\begin{rem}
It is interesting to discuss the relations among the original problem
(P2), the relaxed problem (P2.2), and the restricted problem (P2.5).
Let $\Pi_{1}$, $\Pi_{2}$, and $\Pi_{3}$ denote the feasible regions
of $\boldsymbol{W}$, $\boldsymbol{S}$, and $\eta$ to (P2), (P2.2),
and (P2.5), respectively. It is noting that problem (P2.2) is a relaxed
version of (P2) obtained by dropping the rank constraint, and problem
(P2.5) is a restricted version of problem (P2.2). Thus, $\Pi_{1}$
is a subset of $\Pi_{2}$, and $\Pi_{3}$ is a different subset of
$\Pi_{2}$. Therefore, the constructed rank-one solution $\boldsymbol{W}^{\star},\boldsymbol{S}^{\star},\textrm{ and }\eta^{\star}$
in Proposition \ref{prop:rank-one construction for prob error} are
not necessarily feasible to the corresponding restricted problem (P2.5),
though it is ensured to be feasible for (P2). However, it is interesting
to find that $\boldsymbol{W}^{\star},\boldsymbol{S}^{\star},\textrm{ and }\eta^{\star}$
in Proposition \ref{prop:rank-one construction for prob error} are
always feasible to problem (P2.5) in our simulations. 
\end{rem}
\begin{rem}
By comparing Proposition \ref{prop:1} and Proposition \ref{prop:rank-one construction for prob error},
it is interesting to observe that similar construction methods are
adopted to construct the desired rank-one solutions, for which the
intuitions can be explained as follows. In particular, as the BS has
the perfect CSI of the CU in both scenarios, we can always project
the general-rank information beamforming matrix $\tilde{\boldsymbol{W}}^{*}$
or $\tilde{\boldsymbol{W}}^{\star}$ into the subspace of the CU's
channel (see (\ref{eq:construction W}) or (\ref{eq:construction W new})),
which will not affect the SINR at the CU. At the same time, we allocate
the remaining power after projection to the sensing/AN beams $\boldsymbol{S}^{*}$
or $\boldsymbol{S}^{\star}$ (see (\ref{eq:construction S}) or (\ref{eq:construction S new}))
such that the SINR at all the eavesdroppers would be degraded, thus
improving the secure communication performance. As a result, based
on such construction methods, we obtain the desired rank-one information
beamforming solution for both problems (P1) and (P2). 
\end{rem}

\section{Numerical Results}

This section provides numerical results to validate the performance
of our proposed joint information and sensing beamforming designs
for secure \ac{isac}. In the simulations, the BS is equipped with
$N=8$ antennas. We set the noise powers at the \ac{cu} and all
the eavesdroppers to be identical as $-50\textrm{ dBm},$ i.e., $\sigma_{0}^{2}=\sigma_{k}^{2}=-50\textrm{ dBm},\forall k\in\mathcal{K}_{\mathrm{E}}$
\cite{hua2021optimal}. We assume that the signal attenuation from
the BS to all the targets are identical to be $\varphi=50\mathrm{\textrm{ dB}}$
corresponding to an equal distance of 5 meters, and the signal attenuation
from the BS to the CU to be $\alpha=70\textrm{ dB}$ corresponding
to a distance of 50 meters. Furthermore, the transmit power budget
at each antenna of the BS is set to be $Q=\frac{1}{N}\textrm{ Watt}$
such that the total transmit power at the BS is 1 Watt or 30 dBm.
We choose $M=500$ angles for $\{\bar{\theta}_{m}\}$, which are uniformly
sampled over $[-\frac{\pi}{2},\frac{\pi}{2}]$. We also set the width
of beampattern angle as $\Delta\theta=10{^\circ}$. The normalized
spacing between two adjacent antennas is set as $\frac{d}{\lambda}=0.5.$
There are $K=4$ targets located at angles $-15{^{\circ}}$, $15{^\circ}$,
$-45{^\circ}$, and $45{^\circ}$, i.e., $\{\theta_{k}\}=\{\hat{\theta_{k}}\}=\{-15{^{\circ}},15{^{\circ}},-45{^{\circ}},45{^{\circ}}\}$,
among which the first $K_{\mathrm{E}}=2$ targets (located at $-15{^{\circ}}$
and $15{^{\circ}}$) are untrusted targets. Furthermore, we consider
the \ac{los} channel from the BS to the CU, i.e., $\boldsymbol{g}=\sqrt{10^{-7}}\boldsymbol{a}(\theta_{0})$.
The weight of cross-correlation pattern is set as $\omega_{C}=1$.
For performance comparison, we consider the separate information and
sensing beamforming design as a benchmark scheme as follows. 
\begin{itemize}
\item Separate design: The information and sensing beamformers are separately
designed in the following two stages. First, we obtain the information
beamforming vector $\boldsymbol{w}_{0}$ or $\boldsymbol{W}=\boldsymbol{w}_{0}\boldsymbol{w}_{0}^{H}$
to minimize the transmit power $\textrm{Tr}(\boldsymbol{W})$ while
ensuring the secrecy requirements for bounded CSI errors or Gaussian
errors (i.e., (\ref{eq:worst-case secrecy rate constraint}) and (\ref{eq:Secrecy outage})),
respectively. These problems have been well investigated in literature
\cite{huang2011robust,MaHongJ14}. Denote the obtained information
beamforming covariance solution as $\bar{\boldsymbol{W}}^{\star}$.
Next, with the obtained $\bar{\boldsymbol{W}}^{\star}$, we design
the dedicated sensing signal by restricting it in the null space of
the CU channel to guarantee the secrecy requirements. Towards this
end, we set $\boldsymbol{S}=\boldsymbol{Q}_{2}\bar{\boldsymbol{S}_{2}}\boldsymbol{Q}_{2}^{H}$,
where $\boldsymbol{Q}_{2}=\boldsymbol{I}-\boldsymbol{g}\boldsymbol{g}^{H}/\|\boldsymbol{g}\|^{2}$.
Then, we optimize the sensing beampattern subject to a power constraint
as
\begin{eqnarray}
 & \underset{\bar{\boldsymbol{S}_{2}}\succeq\boldsymbol{0},\eta}{\min} & \frac{1}{M}\sum_{m=1}^{M}\big|\eta\hat{P}(\bar{\theta}_{m})-\boldsymbol{a}^{H}(\bar{\theta}_{m})(\bar{\boldsymbol{W}}^{\star}+\boldsymbol{Q}_{2}\bar{\boldsymbol{S}_{2}}\boldsymbol{Q}_{2}^{H})\boldsymbol{a}(\bar{\theta}_{m})\big|^{2}\label{eq:sep beampattern optimization}\\
 &  & +\frac{2\omega_{C}}{K^{2}-K}\sum_{p=1}^{K}\sum_{q=p+1}^{K}\big|\boldsymbol{a}^{H}(\hat{\theta_{p}})(\bar{\boldsymbol{W}}^{\star}+\boldsymbol{Q}_{2}\bar{\boldsymbol{S}_{2}}\boldsymbol{Q}_{2}^{H})\boldsymbol{a}(\hat{\theta_{q}})\big|^{2}\nonumber \\
 & \textrm{s.t.} & \mathrm{Tr}(\boldsymbol{E}_{i}\boldsymbol{Q}_{2}\bar{\boldsymbol{S}_{2}}\boldsymbol{Q}_{2}^{H}+\boldsymbol{E}_{i}\bar{\boldsymbol{W}}^{\star})=Q,\forall i\in\{1,\dots,N\}.\nonumber 
\end{eqnarray}
Problem (\ref{eq:sep beampattern optimization}) is a convex optimization
problem that can be solved via CVX, for which the optimal solution
is obtained by $\bar{\boldsymbol{S}_{2}}^{\star}$. Accordingly, we
have the sensing beamformers as $\bar{\boldsymbol{S}}^{\star}=\boldsymbol{Q}_{2}\bar{\boldsymbol{S}_{2}}^{\star}\boldsymbol{Q}_{2}^{H}$.
\end{itemize}

In Section V-A, we first evaluate the performance of our proposed
design in the special case with perfect CSI, which corresponds to
(P1) with $\boldsymbol{e}_{k}=\boldsymbol{0}$ or (P2) with $\boldsymbol{C}_{k}=\boldsymbol{0}$,
$\forall k\in\mathcal{K}_{\mathrm{E}}$, and $\rho=0$ (see Remark
\ref{thm:It-is-worth}), as compared to the separate design benchmark.
In Section V-B and Section V-C, we consider the two scenarios with
bounded and Gaussian CSI errors of eavesdroppers, respectively.

\subsection{Special Case with Perfect CSI}

In this subsection, we consider the special case with perfect CSI.
In this case, we assume \ac{los} channels from the BS to all the
potential eavesdroppers with $\boldsymbol{h}_{k}=\sqrt{10^{-3}}\boldsymbol{a}(\theta_{k})$,
$\forall k\in\mathcal{K}_{\mathrm{E}}$. 

First, Fig. \ref{fig:2} shows the achieved transmit beampatterns
by the proposed design (i.e., the joint information and sensing beampattern),
in which the secrecy rate threshold is set as $R_{0}=4\textrm{ bps/Hz}$,
and the CU is located at $\theta_{0}=0{^\circ}$. To reveal more insights
and for comparison, we also show the transmit information beampattern
part (i.e., $\boldsymbol{a}^{H}(\theta)\boldsymbol{W}\boldsymbol{a}(\theta)$
achieved by the information signals $\boldsymbol{W}$ only) and the
dedicated sensing signal beampattern part (i.e., $\boldsymbol{a}^{H}(\theta)\boldsymbol{S}\boldsymbol{a}(\theta)$
achieved by the sensing signals $\boldsymbol{S}$ only) under our
proposed design, as well as the optimal sensing-only beampattern $\boldsymbol{a}^{H}(\theta)\boldsymbol{R}_{\mathrm{sen}}\boldsymbol{a}(\theta)$,
where $\boldsymbol{R}_{\mathrm{sen}}$ is obtained by solving the
sensing beampattern optimization problem (e.g., (P1) with $R_{0}=0\textrm{ bps/Hz}$).
It is observed that the information signal beampatterns are designed
towards the \ac{cu}, while the sensing signal beampatterns are
designed towards both untrusted and trusted targets, thus maximizing
the sensing performance while ensuring the CU's secure communication
requirements. It is also observed that there is a significant discrepancy
between the achieved joint beampattern and the desired beampattern
at the angle of CU ($0{^\circ}$), as the information signal beams
need to be steered towards this direction to satisfy the secure communication
requirement. Furthermore, for the proximate angles of the eavesdropping
targets, the information signal beams are suppressed substantially
to reduce the potential of information leakage, while the energies
of sensing signal beams are focused on these angles to ensure the
sensing requirements and degrade the quality of the eavesdropping
channels.

Fig. \ref{fig:3} shows the sensing beampattern error $D(\boldsymbol{W},\boldsymbol{S},\eta)$
versus the secrecy rate threshold $R_{0}$ with $\theta_{0}=3{^\circ}$.
It is observed that in the whole regime of $R_{0}$, the proposed
optimal design achieves lower sensing beampattern errors than that
of the separate design benchmark scheme. The performance gap is around
20\% of the sensing beampattern error achieved by the sensing-only
benchmark. When $R_{0}$ becomes sufficiently low, the optimal design
approaches the error lower bound obtained by the sensing-only benchmark,
but the separate design performs worse than the optimal design. This
is because for the separate design, the sensing beam is designed within
the null space of the CU channel, thus limiting the available design
DoF for sensing and degrading its performance. 

\begin{figure}
\centering%
\begin{minipage}[t]{0.4\textwidth}%
\centering

\includegraphics[scale=0.48]{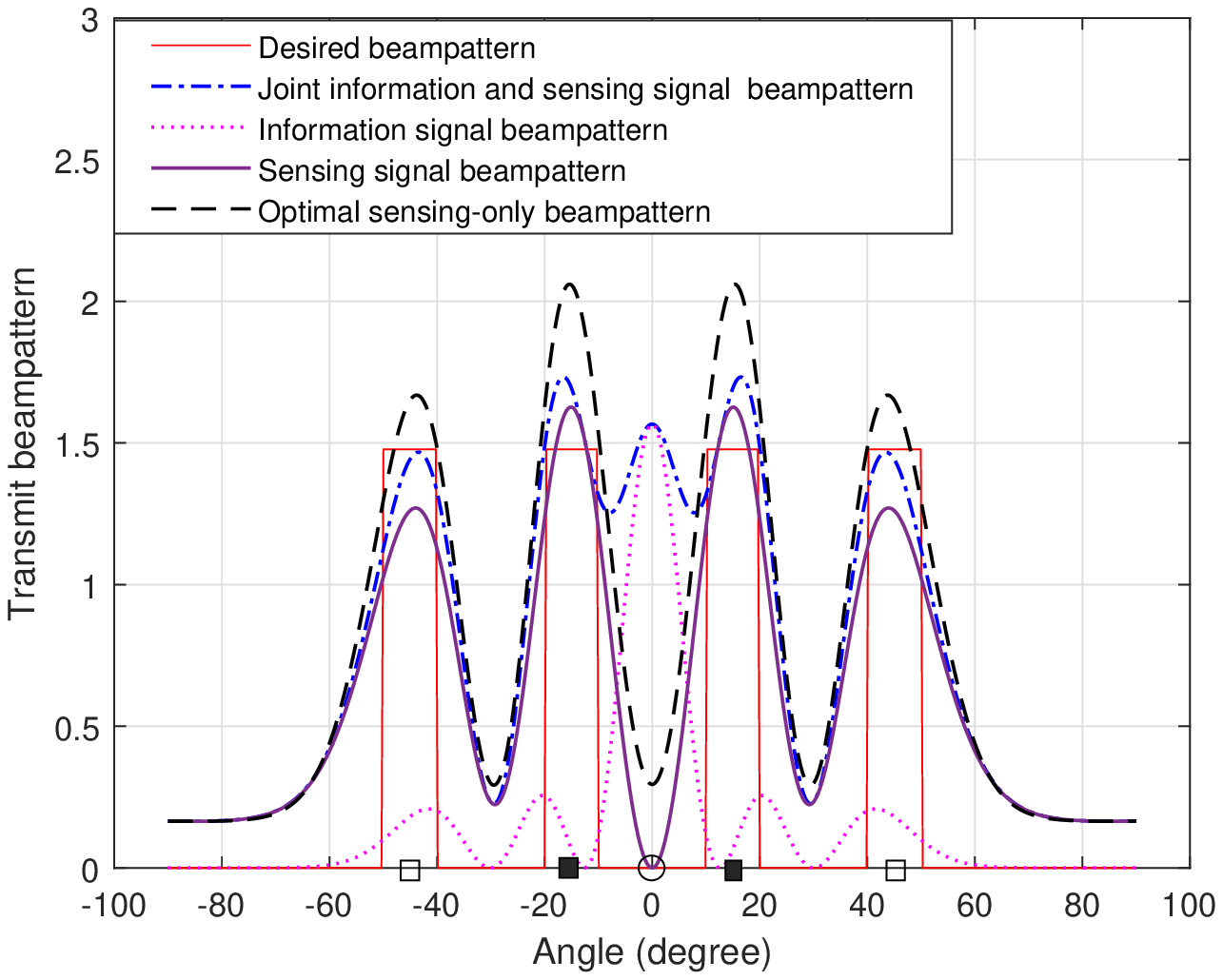}\caption{\label{fig:2}The transmit beampatterns in the special case of perfect
CSI, where $R_{0}=4\textrm{ bps/Hz}$. The circle indicates the CU
angle, the solid rectangles indicate the untrusted target angles,
and the hollow rectangles indicate the trusted ones.}
\end{minipage}\hspace{0.15in}%
\begin{minipage}[t]{0.4\textwidth}%
\centering\includegraphics[scale=0.48]{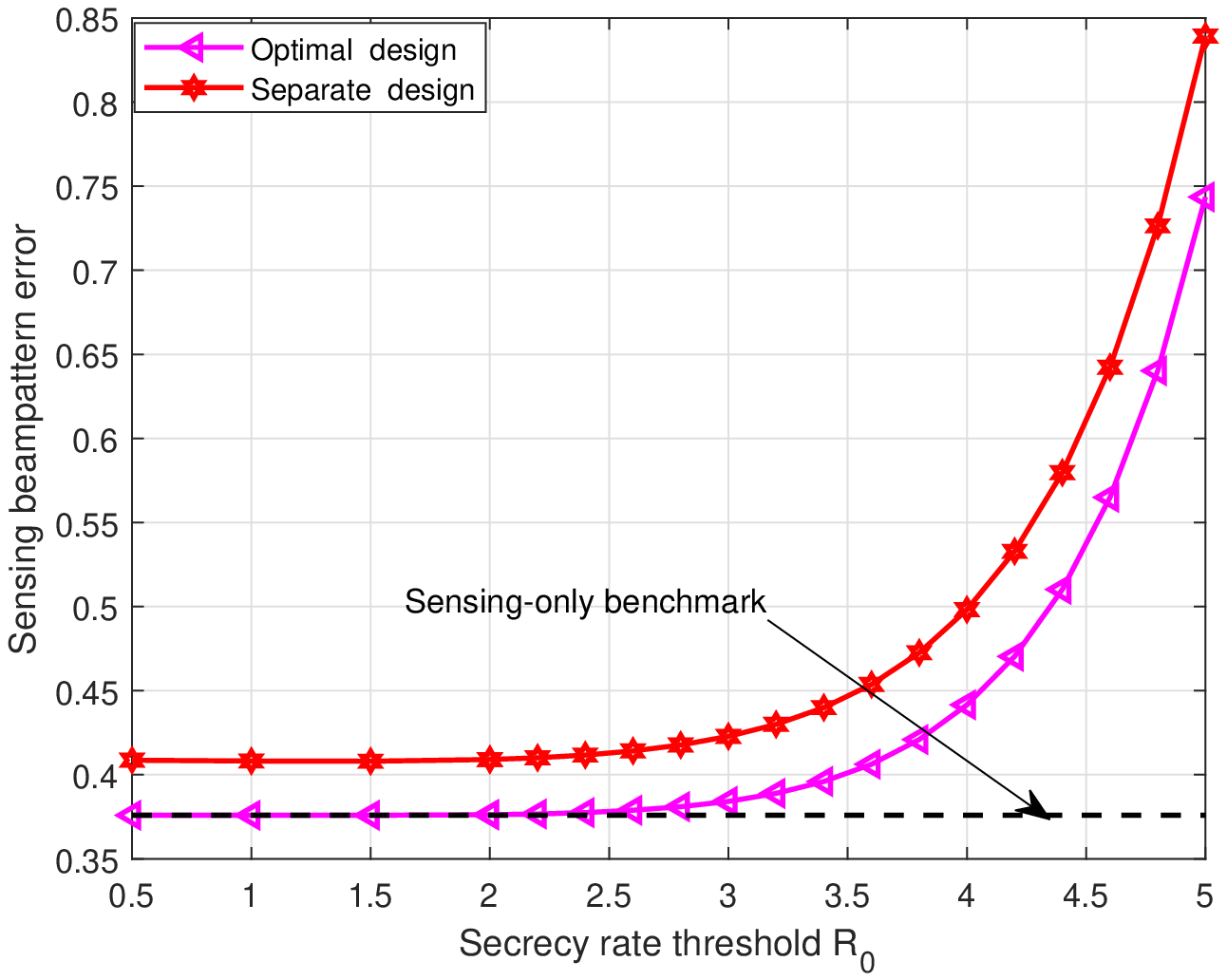}\centering\caption{\label{fig:3}The sensing beampattern error $D(\boldsymbol{W},\boldsymbol{S},\eta)$
versus the secrecy rate threshold $R_{0}$ in the special case of
perfect CSI.}
\end{minipage}
\end{figure}

Fig. \ref{fig:4} shows the sensing beampattern error $D(\boldsymbol{W},\boldsymbol{S},\eta)$
versus the CU angle $\theta_{0}$, with $R_{0}=4\textrm{ bps/Hz}$.
Notice that the secrecy rate constraint $R_{0}=4\textrm{ bps/Hz}$
cannot be satisfied when $\theta_{0}\in(-20{^\circ},-10{^\circ})\cup(10{^\circ},20{^\circ})$,
as the CU is in the neighborhood of the eavesdropping targets such
that no secure information transmission is feasible. As a result,
there is no sensing beampattern error shown in these regions. First,
it is observed that the separate design performs inferior to our proposed
optimal design in term of sensing beampattern error. This is consistent
with the observation in Fig. \ref{fig:3}. When the CU is located
close to a trusted target (i.e., $\theta_{0}=-45{^\circ}$ and $45{^\circ}$),
the sensing beampattern error is observed to be low, as the information
signals play the dual role of communicating with the CU and sensing
the trusted target concurrently. 
\begin{figure}
\centering%
\begin{minipage}[t]{0.4\textwidth}%
\centering\includegraphics[scale=0.48]{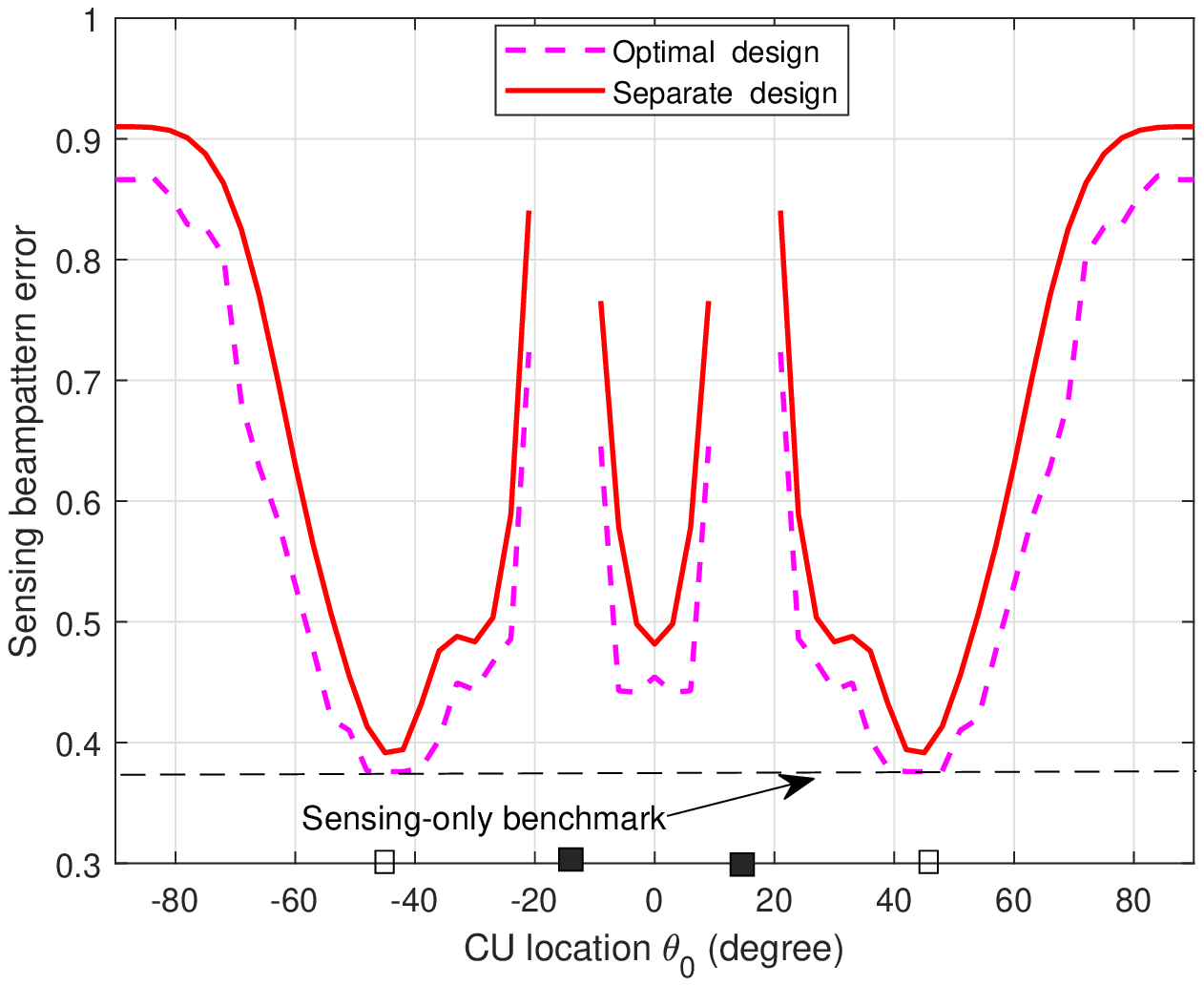}

\caption{\label{fig:4}The sensing beampattern error $D(\boldsymbol{W},\boldsymbol{S},\eta)$
versus the CU angle $\theta_{0}$ in the special case of perfect CSI.
The solid rectangles indicate untrusted target angles, and the hollow
rectangles indicate trusted ones.}
\end{minipage}\hspace{0.15in}%
\begin{minipage}[t]{0.4\textwidth}%
\centering

\includegraphics[scale=0.48]{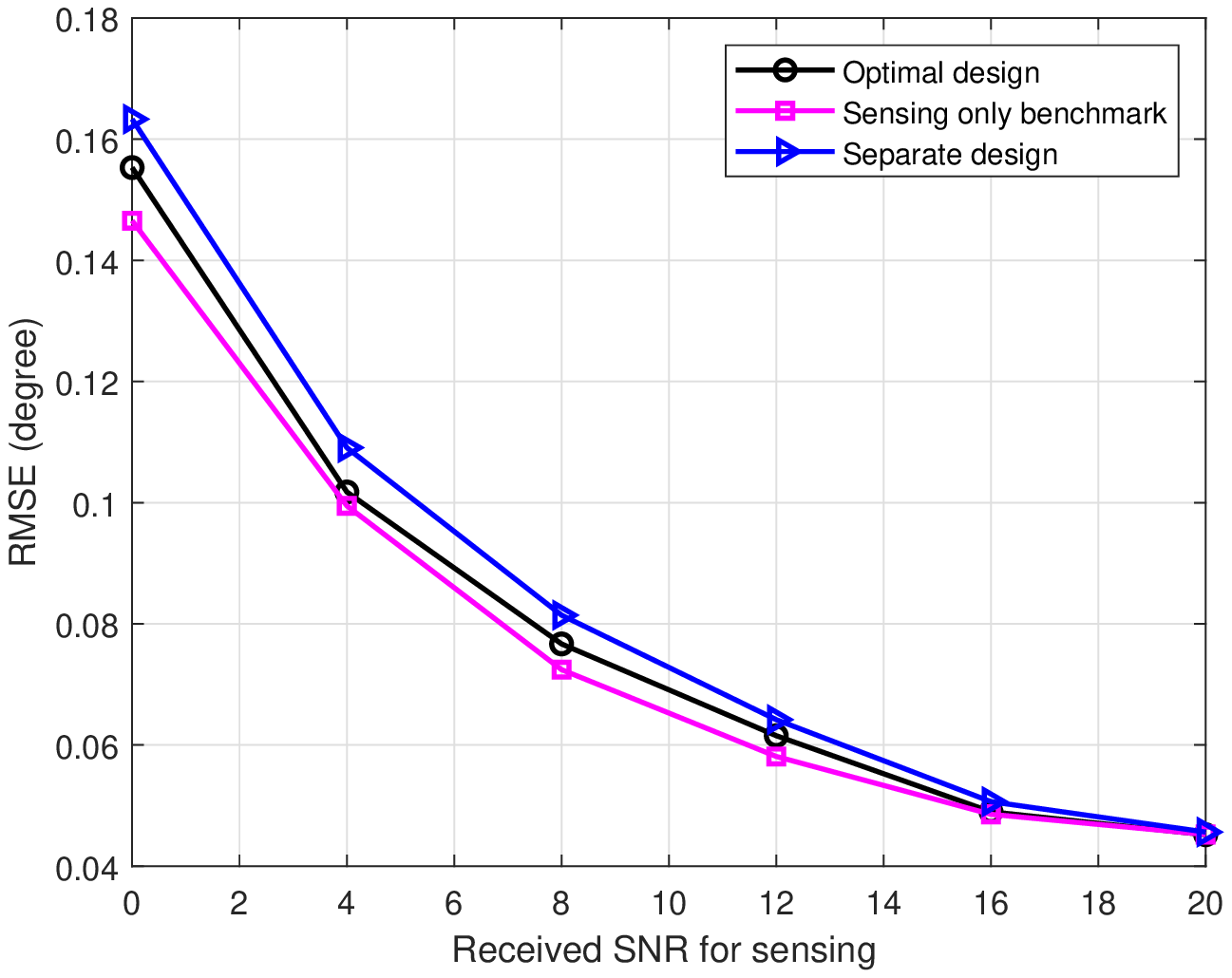}\centering\caption{\label{fig:6}RMSE versus the received SNR for sensing with perfect
CSI.}
\end{minipage}
\end{figure}


Furthermore, besides showing the behaviors of the transmit beampatterns,
we also show the actual angle estimation performance of multiple targets
at the BS sensing receiver to validate the efficiency of our proposed
joint transmit beamforming design. In the simulation, we consider
the sequence with a number of $L=256$ symbols for multi-target estimation,
in which the transmit signal $\boldsymbol{x}(n)=\boldsymbol{w}_{0}s_{0}(n)+\boldsymbol{S}^{1/2}\boldsymbol{u}(n)$
in each symbol $n$ is generated based on $\boldsymbol{u}(n)\sim\mathcal{CN}(\boldsymbol{0},\boldsymbol{I})$
and $s_{0}(n)\sim\mathcal{CN}(0,1)$, $n\in\{1,\ldots,L\}$. Also,
we adopt the sample correlation matrix $\boldsymbol{R}_{xx}=\frac{1}{L}\sum_{n=1}^{L}\boldsymbol{x}(n)\boldsymbol{x}^{H}(n)$
to approximate the transmit covariance matrix $\boldsymbol{W}+\boldsymbol{S}$.
It was shown in \cite{StoPETLiJ07} that their difference is quite
small when the sample length $L=256$. Accordingly, the angles $\{\theta_{k}\}$
are estimated based on the received signals in \eqref{eq:Received signal at CU}
via the Capon technique \cite{xu2008target}. In particular, the Capon
spatial spectrum is obtained as 
\begin{equation}
\frac{\big|\boldsymbol{a}^{H}(\theta)\boldsymbol{R}_{yy}^{-1}\boldsymbol{R}_{yx}\boldsymbol{a}(\theta)\big|}{\big(\boldsymbol{a}^{H}(\theta)\boldsymbol{R}_{yy}^{-1}\boldsymbol{a}(\theta)\big)\big(\boldsymbol{a}^{T}(\theta)\boldsymbol{R}_{xx}\boldsymbol{a}^{c}(\theta)\big)},\label{eq:Capon}
\end{equation}
where $\boldsymbol{R}_{yy}=\frac{1}{L}\sum_{n=1}^{L}\boldsymbol{y}(n)\boldsymbol{y}^{H}(n)$,
$\boldsymbol{R}_{xx}=\frac{1}{L}\sum_{n=1}^{L}\boldsymbol{x}(n)\boldsymbol{x}^{H}(n)$,
and $\boldsymbol{R}_{yx}=\frac{1}{L}\sum_{n=1}^{L}\boldsymbol{y}(n)\boldsymbol{x}^{H}(n)$.
In the implementation of the Capon technique, we uniformly divide
$[-\frac{\pi}{2},\frac{\pi}{2}]$ into 1000 grids to generate the
Capon spatial spectrum, and accordingly find the estimated angles
as the $K$ ones with peak spectrum values.

Fig. \ref{fig:6} shows the angle estimation root mean squared error
(RMSE) by the Capon technique versus the received SNR for sensing
at the BS, which is obtained by averaging over 1000 random realizations.
Here, the RMSE is defined as $\mathrm{RMSE}=\sqrt{\frac{1}{K}\sum_{k=1}^{K}(\tilde{\theta}_{k}-\theta_{k})^{2}}$,
where $\{\tilde{\theta}_{k}\}$ denote the estimate of $\{\theta_{k}\}$.
It is observed that the sensing-only benchmark achieves the best sensing
performance. Furthermore, the proposed optimal design is observed
to lead to a lower RMSE than the separate design. This shows that
our design for minimizing the sensing beampattern error is efficient
in enhancing the multi-target angle estimation accuracy, since minimizing
the cross-correlation pattern in (\ref{eq:cross}) ensures the resolution
of different angles and minimizing the beampattern matching MSE in
(\ref{eq:bme}) ensures the energy towards directions of angles.




\begin{figure}
\centering%
\begin{minipage}[t]{0.4\textwidth}%
 \centering\includegraphics[scale=0.48]{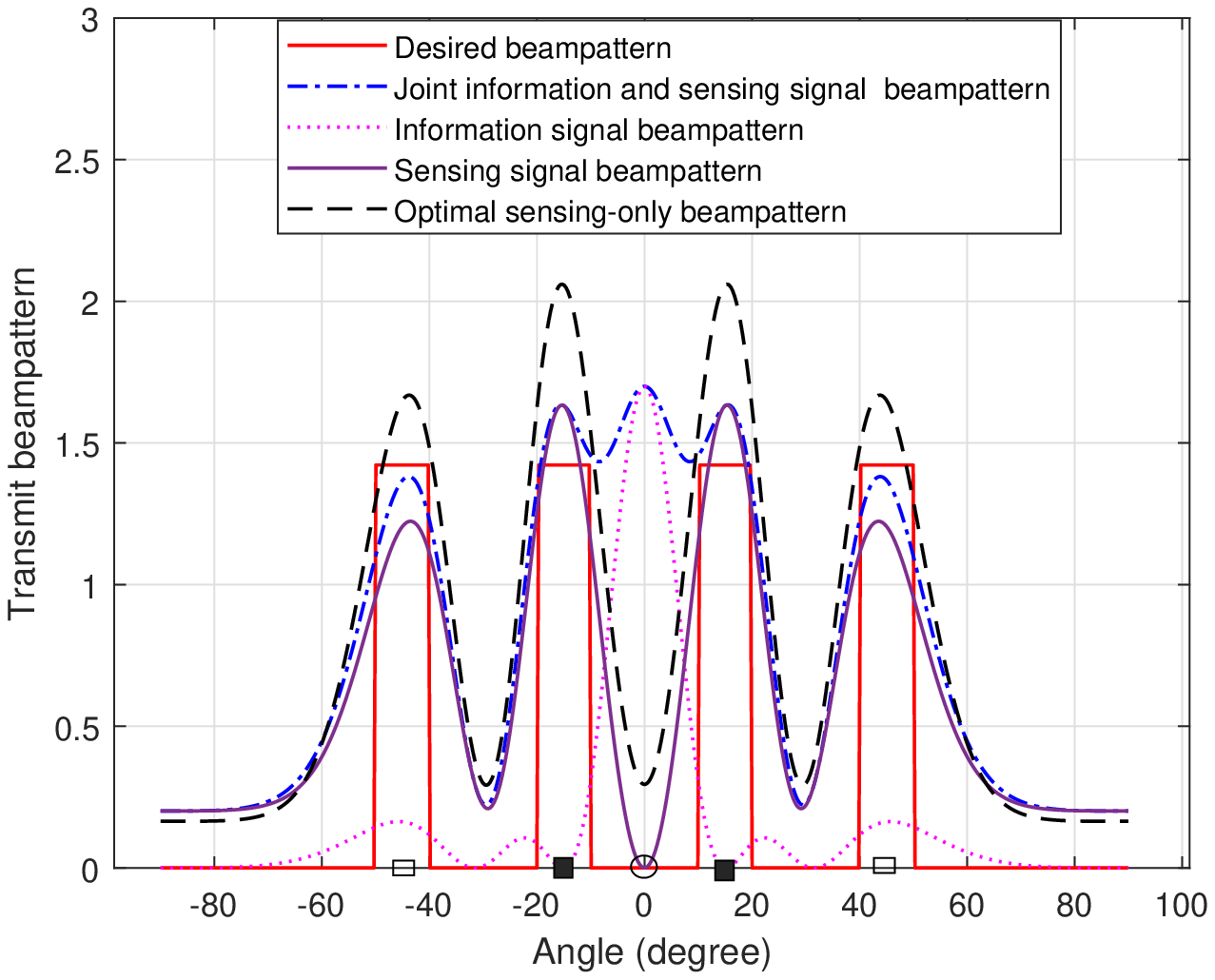}\centering\caption{\label{fig:7}Transmit beampattern with bounded CSI errors, $R_{0}=4\textrm{ bps/Hz}$,
$\mu=0.3$. The circle indicates the direction of the CU and the rectangles
indicate directions of targets with solid rectangles for the untrusted
ones and hollow rectangles for the trusted ones.}
\end{minipage}\hspace{0.15in}%
\begin{minipage}[t]{0.4\textwidth}%
\centering

\includegraphics[scale=0.48]{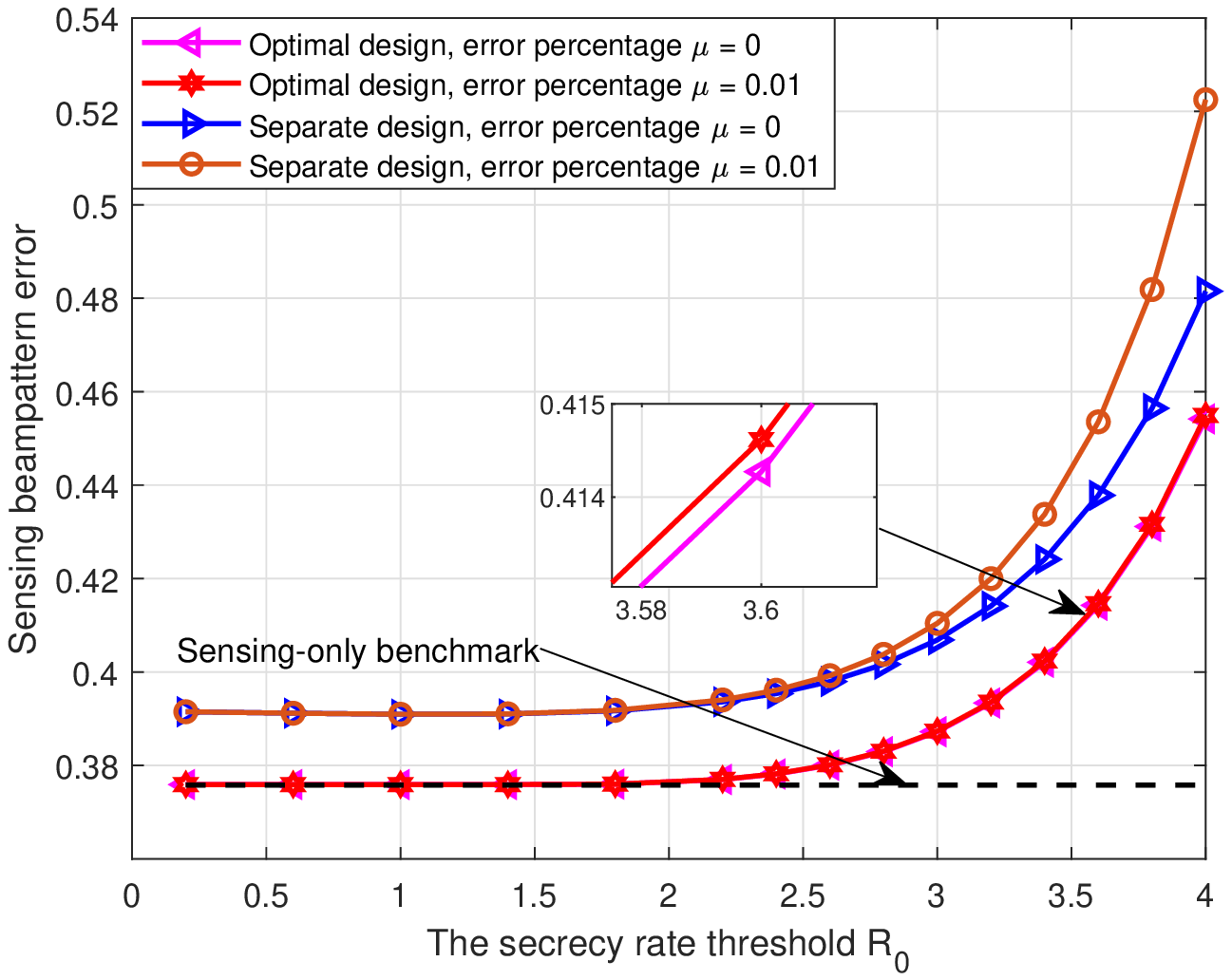}\centering\caption{\label{fig:8}The sensing beampattern error $D(\boldsymbol{W},\boldsymbol{S},\eta)$
versus the worst-case secrecy rate threshold $R_{0}$ with bounded
CSI errors of eavesdroppers.}
\end{minipage}
\end{figure}

\subsection{Scenario with Bounded CSI Errors}

Next, we consider the scenario with bounded CSI errors of eavesdroppers,
in which the estimated channels of each eavesdropping target are set
as $\hat{\boldsymbol{h}_{k}}=\sqrt{10^{-3}}\boldsymbol{a}(\theta_{k})$,
$\forall k\in\mathcal{K}_{\mathrm{E}}$, and the corresponding CSI
error bound is set as $\varepsilon_{k}=\mu|\hat{\boldsymbol{h}_{k}}|$,
$\forall k\in\mathcal{K}_{\mathrm{E}}$, where $\mu$ denotes the
error percentage. The CU is located at angle $\theta_{0}=0{^\circ}$.

Fig. \ref{fig:7} shows the achieved transmit beampattern, where $R_{0}=4\textrm{ bps/Hz}$
and $\mu=0.3$. By comparing it versus the transmit beampattern with
perfect CSI in Fig. \ref{fig:2}, it is observed that the information
beampattern at the CU angle of $0{^\circ}$ becomes narrower for the
purpose of strengthening the confidential information transmission
to combat the potential information leakage caused by the uncertain
eavesdropping channels. Furthermore, the information transmit beampattern
is observed to be further suppressed at proximate angles of these
untrusted targets ($-15\lyxmathsym{\textdegree}$and $15\lyxmathsym{\textdegree}$)
compared to the case of perfect CSI. As such, the transmit beampatterns
around these untrusted targets are observed to be provided by the
sensing signal beampatterns only, which also help further prevent
potential eavesdropping.

Next, we consider the performance comparison of separate design benchmark
and the optimal design in the bounded CSI errors case. Fig. \ref{fig:8}
shows the sensing beampattern error $D(\boldsymbol{S},\boldsymbol{W},\eta)$
versus the secrecy rate threshold $R_{0}$, in which two different
error percentages (i.e., $\mu$ = 0.01 and 0) are considered. It is
observed that the performance gap between the separate design benchmark
and the optimal (robust) design becomes more significant in the bounded
CSI errors case ($\mu$ = 0.01) than that of the case of perfect CSI,
which indicates that the separate design is actually not robust against
errors. This is due to the fact in the separate design, when deriving
the information beamforming vector, we do not consider the AN role
of sensing signals in degrading the quality of eavesdropping channels.
Thus, this leads to inactive secure communication requirement constraints
and more energy is allocated to the information beamformers than that
needed in the optimal design. The redundant power allocation for information
beamforming in the scenario with bounded CSI errors deteriorates the
sensing performance. It is also observed that higher eavesdropping
channel uncertainty (or a larger value of $\mu$) leads to more significant
sensing beampattern errors, as more transmit power is allocated to
ensure secure communication, thus causing more severe beampattern
distortion. 

Then, we evaluate the angle estimation performance of the separate
design and the optimal design via the Capon technique. Fig. \ref{fig:9}
shows the estimation RMSE versus the received SNR for sensing, where
the error percentage $\mu$ is set as 0.02. For comparison, we also
show the sensing performance achieved by the sensing-only benchmark.
It is observed that the sensing performance of the optimal design
is significantly better than the separate design and performs close
to the sensing-only benchmark, especially when the SNR becomes large.
This shows the robustness of our proposed designs to the bounded CSI
errors and is consistent with the observation in Fig. \ref{fig:8}.

\begin{figure}
\centering%
\begin{minipage}[t]{0.4\textwidth}%
\centering\includegraphics[scale=0.48]{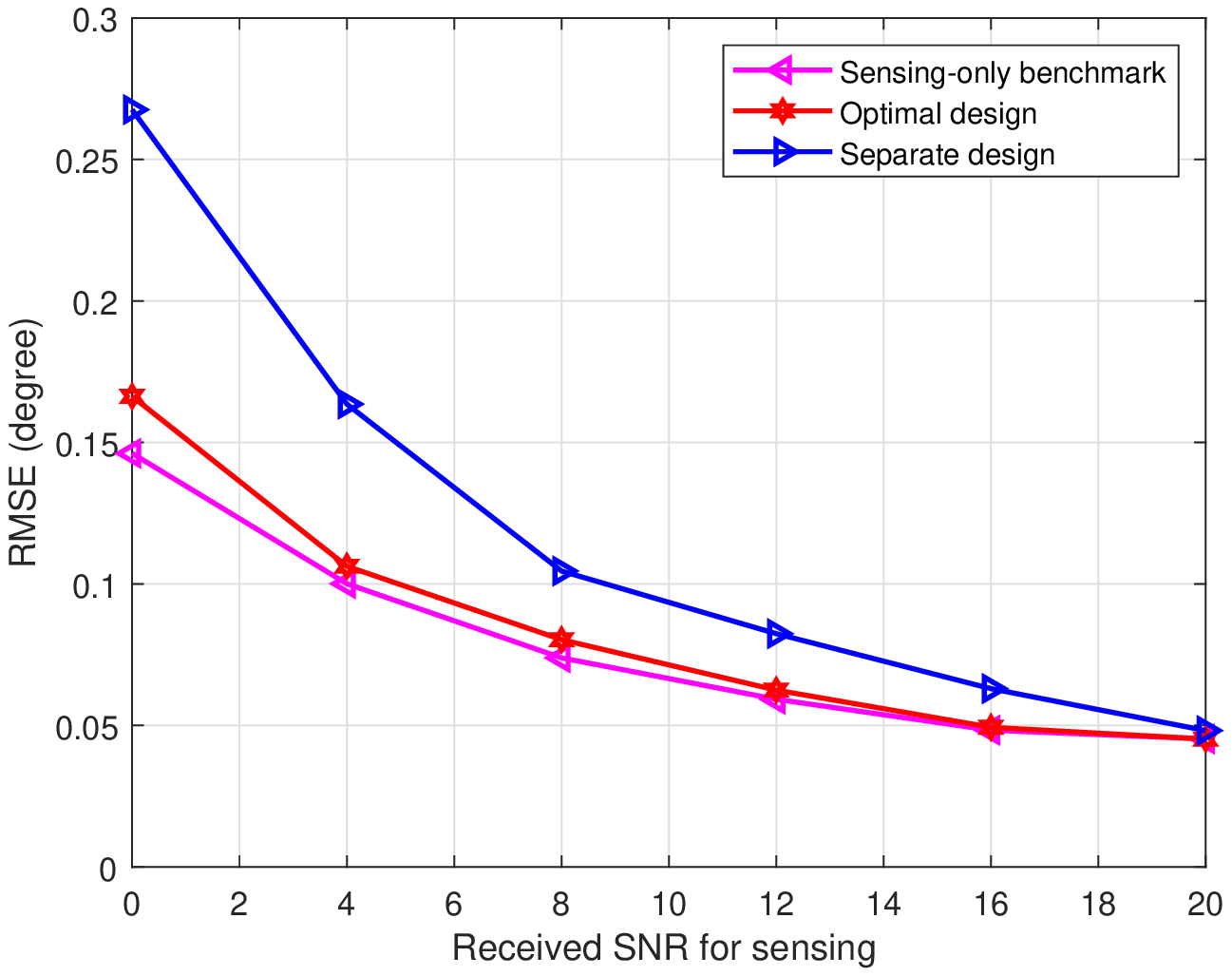}\caption{\label{fig:9}RMSE versus received SNR for sensing with bounded CSI
errors.}
\end{minipage}\hspace{0.15in}%
\begin{minipage}[t]{0.4\textwidth}%
\centering

\includegraphics[scale=0.48]{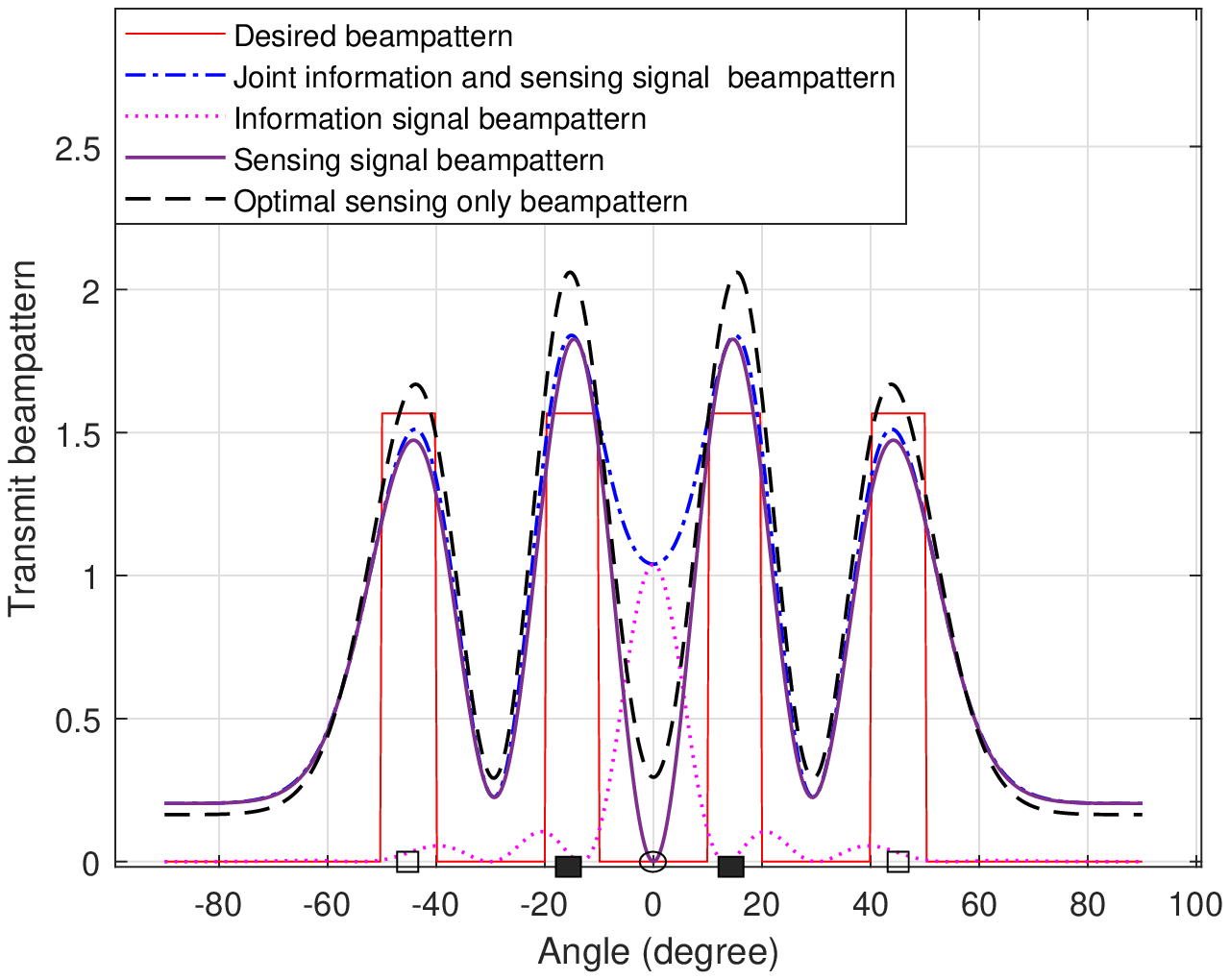}\centering\caption{\label{fig:10}Transmit beampattern with Gaussian CSI errors, $R_{0}=2.5\textrm{ bps/Hz}$,
$\rho=0.1$\textcolor{blue}{.} The circle indicates the direction
of the CU and the rectangles indicate directions of targets with solid
rectangles for the untrusted ones and hollow rectangles for the trusted
ones.}
\end{minipage}
\end{figure}

\subsection{Scenario with Gaussian CSI Errors}

Furthermore, we consider the scenario with Gaussian CSI errors of
eavesdroppers. We consider the Rician fading channels from the BS
to all eavesdroppers, i.e., $\boldsymbol{h}_{k}=\sqrt{\frac{K_{\mathrm{R}}}{K_{\mathrm{R}}+1}}\boldsymbol{h}_{k}^{\mathrm{LoS}}+\sqrt{\frac{1}{K_{\mathrm{R}}+1}}\boldsymbol{h}_{k}^{\mathrm{NLoS}}$
, $\forall k\in\mathcal{K}_{\mathrm{E}}$, where $K_{\mathrm{R}}=10$
is the Rician factor, $\hat{\boldsymbol{h}_{k}}=\sqrt{\frac{K_{\mathrm{R}}}{K_{\mathrm{R}}+1}}\boldsymbol{h}_{k}^{\mathrm{LoS}}=\sqrt{\frac{K_{\mathrm{R}}}{K_{\mathrm{R}}+1}}\sqrt{10^{-3}}\boldsymbol{a}(\theta_{k})$,
and $\boldsymbol{h}_{k}^{\mathrm{NLoS}}\sim\mathcal{CN}(\boldsymbol{0},10^{-3}\boldsymbol{I})$,
$\forall k\in\mathcal{K}_{\mathrm{E}}$. We set the estimated channel
as $\hat{\boldsymbol{h}_{k}}$ and $\boldsymbol{C}_{k}=\frac{1}{K_{\mathrm{R}}+1}10^{-3}\boldsymbol{I}$
for each eavesdropping target $k$. In the case of Gaussian CSI errors,
the separate design benchmark is highly sub-optimal, which may become
infeasible for problem (P2) when the secrecy rate and outage probability
requirements become stringent. Therefore, we do not consider the separate
design in this subsection, but focus on examining the impact of the
secrecy rate threshold $R_{0}$ and the outage threshold $\rho$.

Fig. \ref{fig:10} shows the achieved transmit beampattern with $R_{0}=2.5\textrm{ bps/Hz}$
and the outage threshold $\rho=0.1$. It is observed that the information
signal beam is designed towards the CU angle to ensure secure communication
with Gaussian CSI errors of eavesdroppers, which is similar to the
observation in Fig. \ref{fig:2} with perfect CSI and that in Fig.
\ref{fig:7} with the bounded CSI errors.

Fig. \ref{fig:11} shows the sensing beampattern error $D(\boldsymbol{W},\boldsymbol{S},\eta)$
versus the secrecy rate threshold $R_{0}$, in which $\rho$ is set
as $0.1,0.15,$ and 0.2. It is observed that the sensing beampattern
error increases when $\rho$ becomes small. This is due to the fact
that the secrecy outage constraint becomes stricter in this case,
and more transmit power should be allocated to information signals
towards the CU direction.

Next, Fig. \ref{fig:12} shows the estimation RMSE versus received
SNR for sensing, in which $\rho$ is set as $0.1$ and 0.15. It is
observed that a lower outage threshold of $\rho$ leads to a higher
estimation RMSE. This is because stricter secrecy outage requirement
leads to a higher sensing beampattern error (see Fig. \ref{fig:11})
with more severe beampattern distortion and higher cross-correlation
patterns, thus resulting in worse sensing performance. Our proposed
optimal design is observed to achieve RMSE close to the sensing-only
benchmark. This shows the robustness of our proposed design against
Gaussian CSI errors. 

\begin{figure}
\centering

\begin{minipage}[t]{0.48\textwidth}%
\centering\includegraphics[scale=0.48]{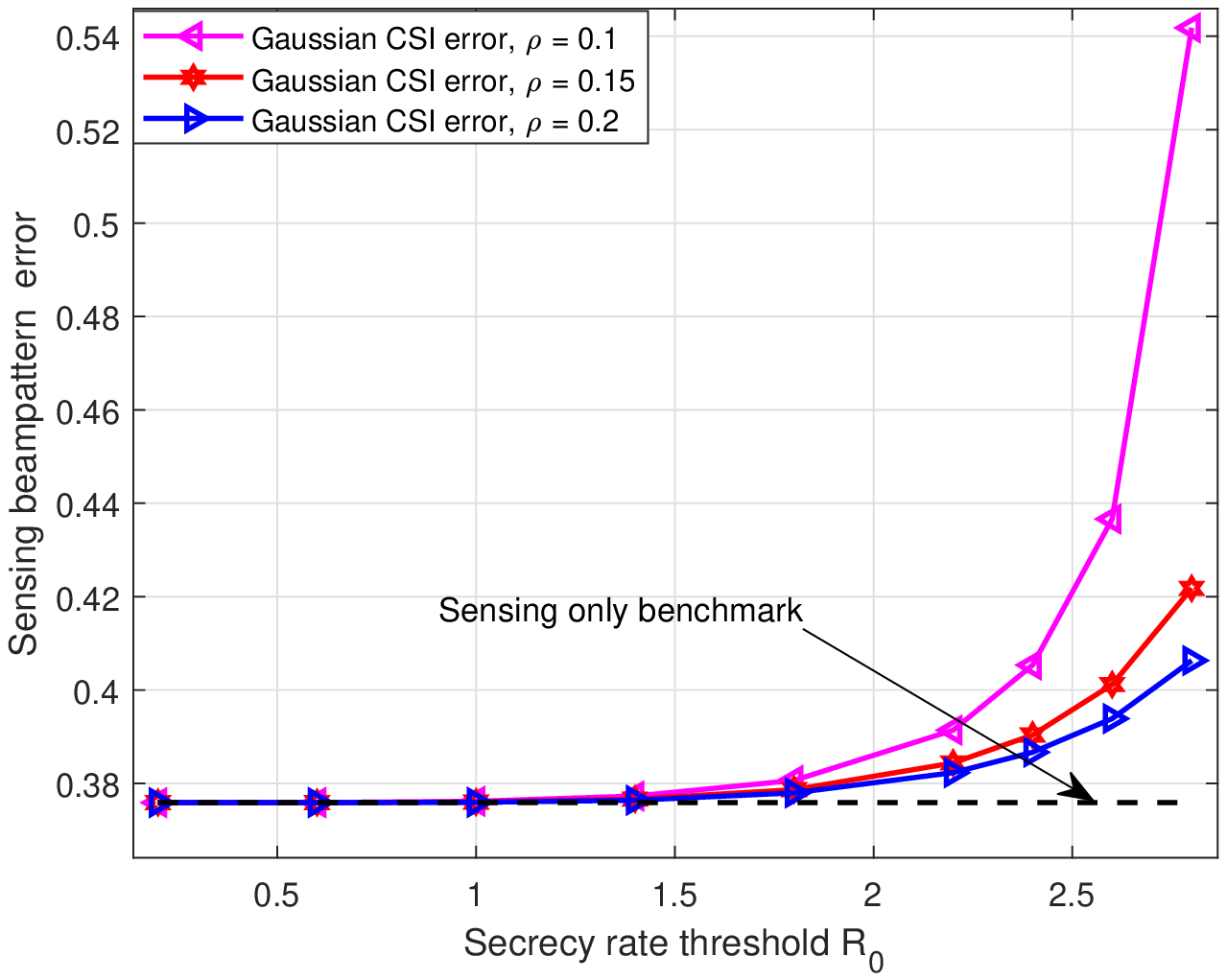}\centering\caption{\label{fig:11}The sensing beampattern error $D(\boldsymbol{W},\boldsymbol{S},\eta)$
versus the secrecy rate threshold $R_{0}$.}
\end{minipage}\hspace{0.15in}%
\begin{minipage}[t]{0.48\textwidth}%
 \centering

\includegraphics[scale=0.48]{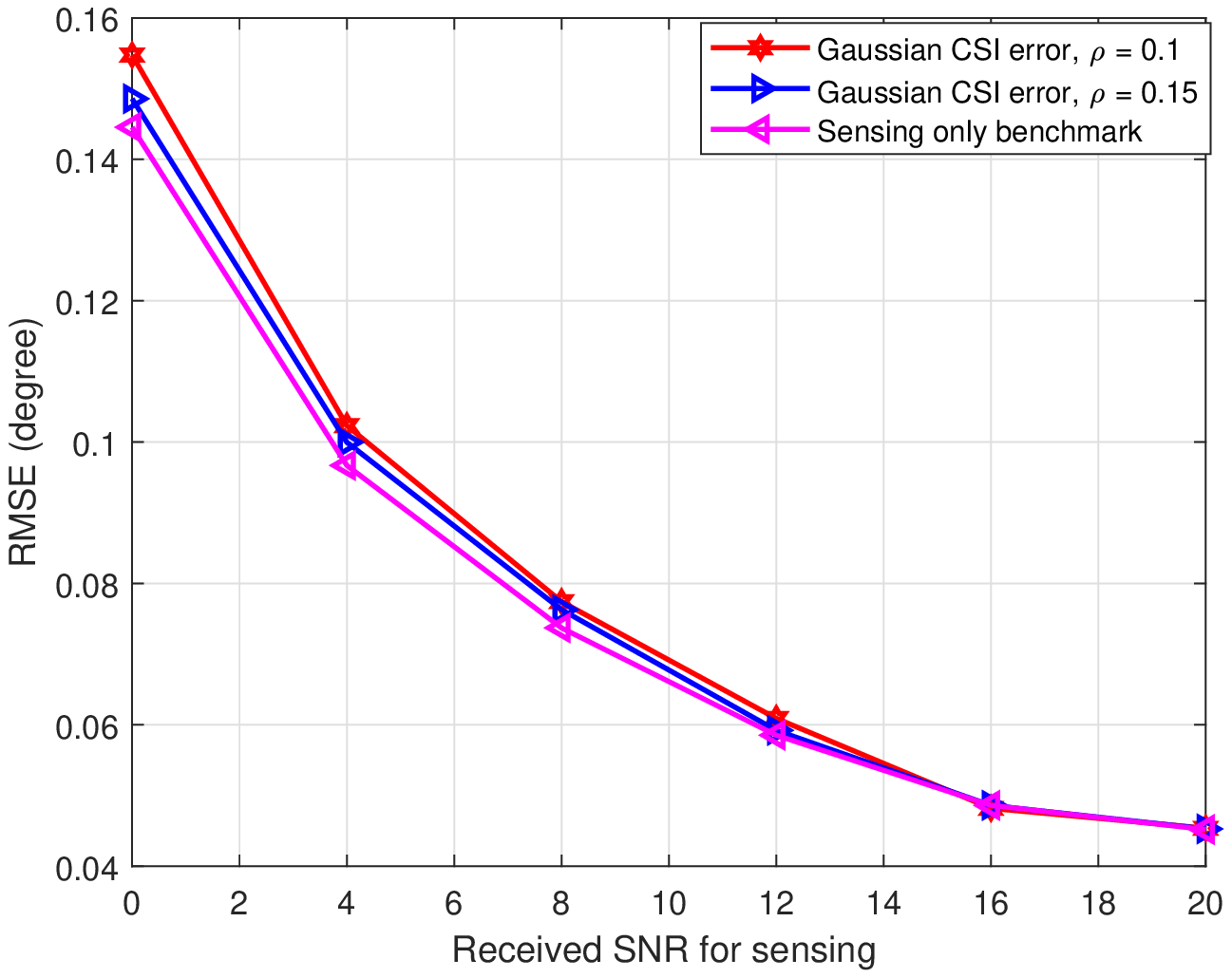}\centering\caption{\label{fig:12}RMSE versus received SNR for sensing with Gaussian
CSI errors.}
\end{minipage}
\end{figure}

\section{Conclusion}

This paper studied the robust transmit beamforming problem for a secure
ISAC system that consists of a BS, a single CU, and multiple sensing
targets. We considered the case of imperfect CSI of eavesdroppers
subject to bounded errors and Gaussian errors, and formulated the
worst-case secrecy rate and the secrecy outage constrained beampattern
optimization problems, respectively. Then, we optimally solved the
worst-case secrecy rate constrained beampattern matching problem by
using the technique of S-procedure, SDR, and 1D search, and rigorously
proved the tightness of the SDR. Next, to solve the secrecy outage
constrained beampattern matching problem, we adopted the safe approximation
based on the BTI to obtain a restricted problem, then solved the restricted
problem using the technique of SDR and 1D search, and finally constructed
a high-quality solution to the original problem. Finally, we provided
numerical results to evaluate the transmit beampattern and the sensing
performance, showing the non-trivial tradeoff between secure communication
and sensing. It was shown that our proposed robust design adjusts
the information and sensing beams to balance the tradeoffs among communicating
with CU, sensing targets, and confusing eavesdroppers, and achieve
desirable sensing transmit beampatterns while ensuring the CU's secrecy
requirements. 

\appendices{}

\section{Proof of Proposition 1}

First, it is observed from (\ref{eq:construction S}) that $\boldsymbol{S}^{*}+\boldsymbol{W}^{*}=\tilde{\boldsymbol{W}}^{*}+\tilde{\boldsymbol{S}}^{*}$.
As a result, $\boldsymbol{W}^{*}$, $\boldsymbol{S}^{*}$, and $\eta^{*}$
achieve the same objective value for problem (P1.3) as that achieved
by $\tilde{\boldsymbol{W}}^{*}$, $\tilde{\boldsymbol{S}}^{*}$, and
$\tilde{\eta}^{*}$. Note that $\boldsymbol{W}^{*}$, $\boldsymbol{S}^{*}$
also satisfy the constraints in (\ref{eq:uniform element power})
and (\ref{eq:semidefinite constraint}).

Next, since $\tilde{\boldsymbol{W}}^{*}\succeq\boldsymbol{0}$, we
have $\tilde{\boldsymbol{W}}^{*}=\boldsymbol{F}\boldsymbol{F}^{H}$.
Based on this, for any $\boldsymbol{v}\in\mathbb{C}^{N\times1}$,
it follows that 
\begin{equation}
\begin{array}[b]{l}
\boldsymbol{v}^{H}(\tilde{\boldsymbol{W}}^{*}-\boldsymbol{W}^{*})\boldsymbol{v}=\boldsymbol{v}^{H}\tilde{\boldsymbol{W}^{*}}\boldsymbol{v}-\boldsymbol{v}^{H}\frac{\tilde{\boldsymbol{W}^{*}}\boldsymbol{g}\boldsymbol{g}^{H}\tilde{\boldsymbol{W}^{*}}}{\boldsymbol{g}^{H}\tilde{\boldsymbol{W}^{*}}\boldsymbol{g}}\boldsymbol{v}\\
=\frac{1}{\boldsymbol{g}^{H}\tilde{\boldsymbol{W}^{*}}\boldsymbol{g}}(\boldsymbol{v}^{H}\tilde{\boldsymbol{W}^{*}}\boldsymbol{v}\boldsymbol{g}^{H}\tilde{\boldsymbol{W}^{*}}\boldsymbol{g}-\boldsymbol{v}^{H}\tilde{\boldsymbol{W}^{*}}\boldsymbol{g}\boldsymbol{g}^{H}\tilde{\boldsymbol{W}^{*}}\boldsymbol{v})\\
=\frac{1}{\boldsymbol{g}^{H}\tilde{\boldsymbol{W}^{*}}\boldsymbol{g}}(\|\boldsymbol{a}\|^{2}\|\boldsymbol{b}\|^{2}-\big|\boldsymbol{a}^{H}\boldsymbol{b}\big|^{2})\overset{\textrm{(a)}}{\geq}0,
\end{array}
\end{equation}
where $\boldsymbol{a}=\boldsymbol{F}^{H}\boldsymbol{v}\in\mathbb{C}^{N\times1},\boldsymbol{b}=\boldsymbol{F}^{H}\boldsymbol{g}\in\mathbb{C}^{N\times1}$,
and inequality (a) holds because of the Cauchy-Schwartz inequality.
Accordingly, we have $\tilde{\boldsymbol{W}}^{*}-\boldsymbol{W}^{*}\succeq\boldsymbol{0}$.
By using this together with (\ref{eq:construction S}), we have $\boldsymbol{S}^{*}\succeq\tilde{\boldsymbol{S}}^{*}\succeq\boldsymbol{0}$.

Furthermore, we prove that $\boldsymbol{W}^{*}$, $\boldsymbol{S}^{*}$,
$\gamma^{*}$, $\eta^{*}$, and $\{\lambda_{k}^{*}\}$ satisfy the
constraints in (\ref{eq:CU constraint}) and (\ref{eq:LMI constriant}).
On the one hand, it is clear from (\ref{eq:construction W}) and (\ref{eq:construction S})
that $\boldsymbol{g}^{H}\tilde{\boldsymbol{W}^{*}}\boldsymbol{g}=\boldsymbol{g}^{H}\boldsymbol{W}^{*}\boldsymbol{g}$
and $\boldsymbol{g}^{H}\tilde{\boldsymbol{S}}^{*}\boldsymbol{g}=\boldsymbol{g}^{H}\boldsymbol{S}{}^{*}\boldsymbol{g}$,
and therefore constraint (\ref{eq:CU constraint}) holds. On the other
hand, as $\tilde{\boldsymbol{W}}^{*}$, $\tilde{\boldsymbol{S}}^{*}$,
$\tilde{\gamma}^{*}$, $\tilde{\eta}^{*}$, and $\{\tilde{\lambda_{k}}^{*}\}$
satisfy the constraint in (\ref{eq:LMI constriant}), we have 
\begin{eqnarray*}
 & \left[\begin{array}{cc}
\tilde{\lambda_{k}}^{*}\boldsymbol{I}-\big(\tilde{\boldsymbol{W}}^{*}-\varphi(\tilde{\gamma}^{*})\tilde{\boldsymbol{S}}^{*}\big) & -\big(\tilde{\boldsymbol{W}}^{*}-\varphi(\tilde{\gamma}^{*})\tilde{\boldsymbol{S}}^{*}\big)\hat{\boldsymbol{h}}_{k}\\
-\hat{\boldsymbol{h}}_{k}^{H}\big(\tilde{\boldsymbol{W}}^{*}-\varphi(\tilde{\gamma}^{*})\tilde{\boldsymbol{S}}^{*}\big) & -\tilde{\lambda_{k}}^{*}\varepsilon_{k}^{2}-\hat{\boldsymbol{h}}_{k}^{H}\big(\tilde{\boldsymbol{W}}^{*}-\varphi(\tilde{\gamma}^{*})\tilde{\boldsymbol{S}}^{*}\big)\hat{\boldsymbol{h}}_{k}+\varphi(\tilde{\gamma}^{*})\sigma_{k}^{2}
\end{array}\right]\succeq\boldsymbol{0}, & \forall k\in\mathcal{K}_{\mathrm{E}}.
\end{eqnarray*}
It follows that 
\begin{eqnarray}
 &  & \left[\begin{array}{cc}
\lambda_{k}^{*}\boldsymbol{I}-\big(\boldsymbol{W}^{*}-\varphi(\gamma^{*})\boldsymbol{S}^{*}\big) & -\big(\boldsymbol{W}^{*}-\varphi(\gamma^{*})\boldsymbol{S}^{*}\big)\hat{\boldsymbol{h}}_{k}\\
-\hat{\boldsymbol{h}}_{k}^{H}\big(\boldsymbol{W}^{*}-\varphi(\gamma^{*})\boldsymbol{S}^{*}\big) & -\lambda_{k}^{*}\varepsilon_{k}^{2}-\hat{\boldsymbol{h}}_{k}^{H}\big(\boldsymbol{W}^{*}-\varphi(\gamma^{*})\boldsymbol{S}^{*}\big)\hat{\boldsymbol{h}}_{k}+\varphi(\gamma^{*})\sigma_{k}^{2}
\end{array}\right]\nonumber \\
 & - & \left[\begin{array}{cc}
\tilde{\lambda_{k}}^{*}\boldsymbol{I}-\big(\tilde{\boldsymbol{W}}^{*}-\varphi(\tilde{\gamma}^{*})\tilde{\boldsymbol{S}}^{*}\big) & -\big(\tilde{\boldsymbol{W}}^{*}-\varphi(\tilde{\gamma}^{*})\tilde{\boldsymbol{S}}^{*}\big)\hat{\boldsymbol{h}}_{k}\\
-\hat{\boldsymbol{h}}_{k}^{H}\big(\tilde{\boldsymbol{W}}^{*}-\varphi(\tilde{\gamma}^{*})\tilde{\boldsymbol{S}}^{*}\big) & -\tilde{\lambda_{k}}^{*}\varepsilon_{k}^{2}-\hat{\boldsymbol{h}}_{k}^{H}\big(\tilde{\boldsymbol{W}}^{*}-\varphi(\tilde{\gamma}^{*})\tilde{\boldsymbol{S}}^{*}\big)\hat{\boldsymbol{h}}_{k}+\varphi(\tilde{\gamma}^{*})\sigma_{k}^{2}
\end{array}\right]\nonumber \\
 & = & \left[\begin{array}{cc}
\Delta\boldsymbol{W}-\varphi(\gamma^{*})\Delta\boldsymbol{S} & \big(\Delta\boldsymbol{W}-\varphi(\gamma^{*})\Delta\boldsymbol{S}\big)\hat{\boldsymbol{h}}_{k}\\
\hat{\boldsymbol{h}}_{k}^{H}\big(\Delta\boldsymbol{W}-\varphi(\gamma^{*})\Delta\boldsymbol{S}\big) & \hat{\boldsymbol{h}}_{k}^{H}\big(\Delta\boldsymbol{W}-\varphi(\gamma^{*})\Delta\boldsymbol{S}\big)\hat{\boldsymbol{h}}_{k}
\end{array}\right],\forall k\in\mathcal{K}_{\mathrm{E}},\label{eq:difference matrix}
\end{eqnarray}
where $\Delta\boldsymbol{W}=\tilde{\boldsymbol{W}}^{*}-\boldsymbol{W}^{*}\succeq\boldsymbol{0}$,
$\Delta\boldsymbol{S}=\tilde{\boldsymbol{S}}^{*}-\boldsymbol{S}^{*}\preceq\boldsymbol{0}$,
and thus $\Delta\boldsymbol{W}-\varphi(\gamma^{*})\Delta\boldsymbol{S}\succeq\boldsymbol{0}$.
Therefore, we only need to prove that the right-hand-side (RHS) of
(\ref{eq:difference matrix}) is positive semi-definite to satisfy
(\ref{eq:LMI constriant}). Towards this end, we introduce the following
theorem about positive semi-definete block matrices. 
\begin{thm}
\label{thm:positive matrix}\cite{horn2012matrix} Let $\boldsymbol{H}=\left[\begin{array}{cc}
\boldsymbol{A} & \boldsymbol{B}\\
\boldsymbol{B}^{H} & \boldsymbol{C}
\end{array}\right]\in\mathbb{S}^{L+T}$ , with $\boldsymbol{A}\in\mathbb{S}^{L}$ and $\boldsymbol{C}\in\mathbb{S}^{T}$.
The following two statements are equivalent:

(a) $\boldsymbol{H}$ is positive semi-definite.

(b) $\boldsymbol{A}$ and $\boldsymbol{C}$ are positive semi-definite,
and there is a contraction $\boldsymbol{X}\in\mathbb{C}^{L\times T}$
such that $\boldsymbol{B}=\boldsymbol{A}^{1/2}\boldsymbol{X}\boldsymbol{C}^{1/2}$. 
\end{thm}
A given matrix is a contraction if its largest singular value is less
than or equal to 1. Recall that $\Delta\boldsymbol{W}-\varphi(\gamma^{*})\Delta\boldsymbol{S}\succeq\boldsymbol{0}$,
we have $\hat{\boldsymbol{h}}_{k}^{H}\big(\Delta\boldsymbol{W}-\varphi(\gamma^{*})\Delta\boldsymbol{S}\big)\hat{\boldsymbol{h}}_{k}\geq0$.
Here, notice that $\Big(\hat{\boldsymbol{h}}_{k}^{H}\big(\Delta\boldsymbol{W}-\varphi(\gamma^{*})\Delta\boldsymbol{S}\big)\hat{\boldsymbol{h}}_{k}\Big)^{1/2}=\|\big(\Delta\boldsymbol{W}-\varphi(\gamma^{*})\Delta\boldsymbol{S}\big)^{1/2}\hat{\boldsymbol{h}}_{k}\|$.
Then, we select 
\begin{equation}
\boldsymbol{a}=\frac{1}{\|\big(\Delta\boldsymbol{W}-\varphi(\gamma^{*})\Delta\boldsymbol{S}\big)^{1/2}\hat{\boldsymbol{h}}_{k}\|}\cdot\big(\Delta\boldsymbol{W}-\varphi(\gamma^{*})\Delta\boldsymbol{S}\big)^{1/2}\hat{\boldsymbol{h}}_{k}.
\end{equation}
Due to $\|\boldsymbol{a}\|=1$, $\boldsymbol{a}$ is a contraction.
Thus, we have 
\begin{equation}
\big(\Delta\boldsymbol{W}-\varphi(\gamma^{*})\Delta\boldsymbol{S}\big)^{1/2}\boldsymbol{a}\|\big(\Delta\boldsymbol{W}-\varphi(\gamma^{*})\Delta\boldsymbol{S}\big)^{1/2}\hat{\boldsymbol{h}}_{k}\|=\big(\Delta\boldsymbol{W}-\varphi(\gamma^{*})\Delta\boldsymbol{S}\big)\hat{\boldsymbol{h}}_{k}.\label{eq:psdm}
\end{equation}
By combining $\Delta\boldsymbol{W}-\varphi(\gamma^{*})\Delta\boldsymbol{S}\succeq\boldsymbol{0}$
and (\ref{eq:psdm}), we prove that the RHS of (\ref{eq:difference matrix})
is positive semi-definite. As a result, constraint (\ref{eq:LMI constriant})
is satisfied.

By combining the results above, it is proved that $\boldsymbol{W}^{*}$,
$\boldsymbol{S}^{*}$, $\gamma^{*}$, $\eta^{*}$, and $\{\lambda_{k}^{*}\}$
are optimal for problem (P1.3). Notice that $\textrm{rank}(\boldsymbol{W}^{*})=1$.
This thus completes the proof.

\section{Proof of Proposition \ref{prop:rank-one construction for prob error}}

Based on the similar proof techniques as in Appendix A, we have 
\begin{eqnarray}
 &  & \boldsymbol{S}^{\star}\succeq\tilde{\boldsymbol{S}}^{\star}\succeq\boldsymbol{0},\label{eq:Delta S}\\
 &  & \tilde{\boldsymbol{W}}^{\star}\succeq\boldsymbol{W}^{\star}\succeq\boldsymbol{0},\label{eq:Delta W}\\
 &  & \tilde{\boldsymbol{W}}^{\star}+\tilde{\boldsymbol{S}}^{\star}=\boldsymbol{W}^{\star}+\boldsymbol{S}^{\star}.\label{eq:sum}
\end{eqnarray}
It follows from (\ref{eq:sum}) that $\boldsymbol{W}^{\star},\boldsymbol{S}^{\star},\textrm{ and }\eta^{\star}$
achieve the same objective value for $D(\boldsymbol{W},\boldsymbol{S},\eta)$
as that achieved by $\tilde{\boldsymbol{W}}^{\star},\tilde{\boldsymbol{S}}^{\star},\textrm{ and }\tilde{\eta}^{\star}$.
Therefore, we need to verify that the constructed solution $\boldsymbol{W}^{\star},\boldsymbol{S}^{\star},\textrm{ and }\eta^{\star}$
satisfy (\ref{eq:secrecy non-outage}), (\ref{eq:Power constraint}),
and ($\textrm{\ref{eq:Semi definite condtraint}}$). Based on (\ref{eq:sum}),
$\boldsymbol{g}^{H}\tilde{\boldsymbol{W}^{\star}}\boldsymbol{g}=\boldsymbol{g}^{H}\boldsymbol{W}^{\star}\boldsymbol{g}$,
and $\boldsymbol{g}^{H}\tilde{\boldsymbol{S}}^{*}\boldsymbol{g}=\boldsymbol{g}^{H}\boldsymbol{S}{}^{*}\boldsymbol{g}$,
(\ref{eq:Power constraint}) and (\ref{eq:Semi definite condtraint})
are satisfied. Then, we focus on proving that (\ref{eq:secrecy non-outage})
is satisfied. Based on (\ref{eq:Delta S}) and (\ref{eq:Delta W}),
we have
\begin{equation}
\frac{\boldsymbol{h}_{k}^{H}\boldsymbol{W}^{\star}\boldsymbol{h}_{k}}{\boldsymbol{h}_{k}^{H}\boldsymbol{S}^{\star}\boldsymbol{h}_{k}+\sigma_{k}^{2}}\leq\frac{\boldsymbol{h}_{k}^{H}\tilde{\boldsymbol{W}}^{\star}\boldsymbol{h}_{k}}{\boldsymbol{h}_{k}^{H}\tilde{\boldsymbol{S}}^{\star}\boldsymbol{h}_{k}+\sigma_{k}^{2}},\forall\boldsymbol{h}_{k}\in\mathbb{C}^{N\times1}.
\end{equation}
As a result,
\begin{eqnarray}
 &  & \textrm{Pr}\big(\log_{2}\big(1+\gamma\big)-\underset{k\in\mathcal{K}_{\mathrm{E}}}{\max}\log_{2}\big(1\negthickspace+\negthickspace\frac{\boldsymbol{h}_{k}^{H}\boldsymbol{W}^{\star}\boldsymbol{h}_{k}}{\boldsymbol{h}_{k}^{H}\boldsymbol{S}^{\star}\boldsymbol{h}_{k}+\sigma_{k}^{2}}\big)\geq R_{0}\big)\nonumber \\
 & \geq & \textrm{Pr}\big(\log_{2}\big(1+\gamma\big)-\underset{k\in\mathcal{K}_{\mathrm{E}}}{\max}\log_{2}\big(1\negthickspace+\frac{\boldsymbol{h}_{k}^{H}\tilde{\boldsymbol{W}}^{\star}\boldsymbol{h}_{k}}{\boldsymbol{h}_{k}^{H}\tilde{\boldsymbol{S}}^{\star}\boldsymbol{h}_{k}+\sigma_{k}^{2}}\big)\geq R_{0}\big)\geq1-\rho.
\end{eqnarray}

This completes the proof.

\appendices{}

{\footnotesize{}\bibliographystyle{IEEEtran}
\bibliography{IEEEabrv,IEEEexample,my_ref}
}{\footnotesize\par}

{\footnotesize{}

{\footnotesize{}

{\footnotesize{}

{\footnotesize{}

{\footnotesize{}

{\footnotesize{}

{\footnotesize{}
\ifCLASSOPTIONcaptionsoff \newpage\fi}{\footnotesize\par}

{\footnotesize{}

{\footnotesize{}

{\footnotesize{}

{\footnotesize{}

{\footnotesize{}

{\footnotesize{}

{\footnotesize{}

{\footnotesize{}

{\footnotesize{}

{\footnotesize{}
\end{document}